\documentclass[a4paper,english]{article}
\usepackage[utf8]{inputenc}
\usepackage{algorithm}
\usepackage{algorithmicx}
\usepackage{thm-restate}
\usepackage[noend]{algpseudocode}
\usepackage{fullpage}
\usepackage{amssymb, amsthm, amsmath, mathabx, url}

\usepackage{graphicx,color}
\usepackage{picinpar}
\usepackage{xspace}
\usepackage{hyperref}
\usepackage{rotating}
\usepackage{color}

\usepackage{complexity}

\newcommand{\cog}[1]{\ensuremath{C \left(#1\right)}}
\newcommand{\UPr}{\textsc{Unload}\xspace}
\newcommand{\dUPr}{\textsc{dUnload}\xspace}
\newcommand{\LPr}{\textsc{Load}\xspace}

\newcommand{\statement}[2]{\vskip2ex\noindent{\sffamily\bfseries#1~\ref{#2}.}}
\newcommand{\gravityc}{center}

\newtheorem{definition}{Definition}
\newtheorem{theorem}[definition]{Theorem}
\newtheorem{lemma}[definition]{Lemma}
\newtheorem{corollary}[definition]{Corollary}
\newtheorem{observation}[definition]{\textbf{Observation}}
\newtheorem{claim}[definition]{Claim}

\usepackage{fp}
\usepackage[dvipsnames]{xcolor}
\usepackage{tikz}

\usetikzlibrary{arrows}
\usetikzlibrary{shapes}
\usetikzlibrary{decorations.pathmorphing}
\usetikzlibrary{decorations.pathreplacing}
\usetikzlibrary{decorations.shapes}
\usetikzlibrary{decorations.text}

\usepackage{authblk}
\usepackage[english]{babel}
\date{}
\title{Don't Rock the Boat:\\ Algorithms for Balanced Dynamic Loading and Unloading\footnote{This is the full version of an extended abstract that will appear in the 13th Latin American Theoretical INformatics Symposium (LATIN 2018), April 16--19, 2018.}}

\author[1]{Sándor P.~Fekete}
\author[1]{Sven von Höveling}
\author[2]{Joseph S.~B.~Mitchell\footnote{Work of this author is partially supported by the National Science Foundation (CCF-1526406).}}
\author[1]{Christian Rieck}
\author[1]{Christian Scheffer}
\author[1]{Arne Schmidt}
\author[2]{James R.~Zuber}
\affil[1]{Department of Computer Science, TU Braunschweig, 38106 Braunschweig, Germany.
	\texttt{\{s.fekete,v.sven,c.rieck,c.scheffer,arne.schmidt\}@tu-bs.de}}
\affil[2]{Department of Applied Mathematics and Statistics, Stony Brook University, Stony Brook, NY 11794, USA.
\texttt{joseph.mitchell@stonybrook.edu, zuber139@gmail.com}}

\begin{document}
	\maketitle
	
\begin{abstract}
We consider dynamic loading and unloading problems for heavy geometric objects. The challenge is
to maintain balanced configurations at all times: minimize the maximal motion of the overall
center of gravity. While this problem has been studied from an algorithmic point of view, 
previous work only focuses on balancing the {\em final} center of gravity; we give a variety of
results for computing balanced loading and unloading schemes that minimize the maximal motion 
of the center of gravity during the entire process.

In particular, we consider the one-dimensional case and distinguish between
{\em loading} and {\em unloading}. In the unloading variant, the positions of the intervals
are given, and we search for an optimal unloading order of the intervals. We
prove that the unloading variant is \NP-complete and give a 2.7-approximation
algorithm. In the loading variant, we have to compute both the positions of the
intervals and their loading order. We give optimal approaches for several
variants that model different loading scenarios that may arise, e.g., in the
loading of a transport ship with containers.

	\end{abstract}
\section{Introduction}

Packing a set of objects is a classic challenge that has been
studied from a wide range of angles: how can the
objects be arranged to fit into the container? Packing problems
are important for a large spectrum of practical applications,
such as loading items into a storage space, or containers onto a ship. 
They are also closely related to problems of scheduling and sequencing, in which the issues of limited
space are amplified by including temporal considerations. 

Packing and scheduling are closely intertwined in {\em loading} and {\em unloading}
problems, where the challenge is not just to compute an acceptable {\em final} configuration, but also
the process of {\em dynamically building} this configuration, such that intermediate
states are both achievable and stable. This is highly relavant in the scenario of 
loading and unloading container ships, for which maintaining {\em balance} throughout the process
is crucial.
Balancedness of packing also plays an important role for
other forms of shipping: Mongeau and Bes~\cite{Mongeau2003} showed that displacing the 
center of gravity by less than 75cm in a long-range aircraft may cause, over a 10,000 km flight, 
an additional consumption of 4,000 kg of fuel.

\begin{figure}[h!]
  \begin{center}
       \includegraphics[width=.4\textwidth]{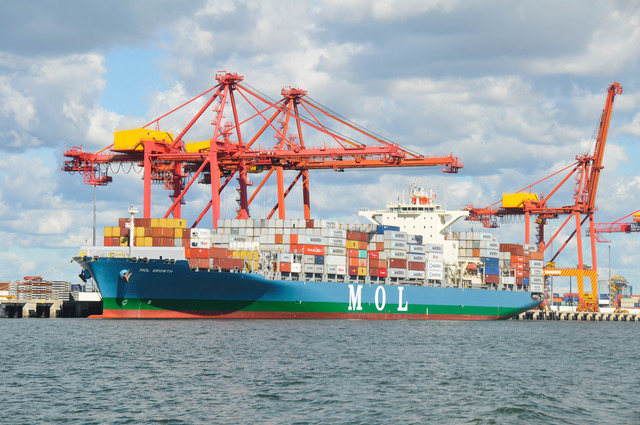}
       \includegraphics[width=.4\textwidth]{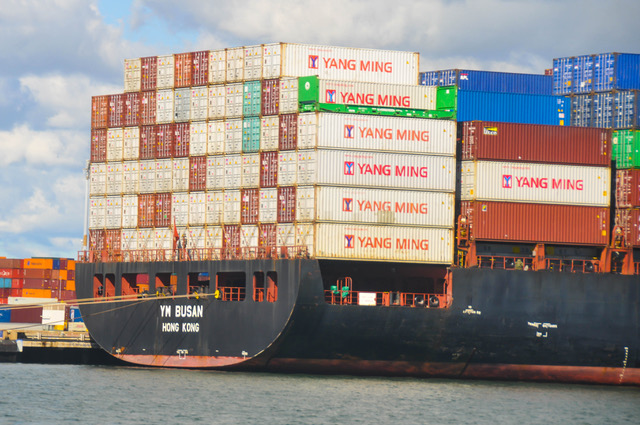}
  \end{center}
\vspace*{-6mm}
  \caption{Loading and unloading container ships.}
  \label{fig:ship}
\end{figure}

In this paper, we consider algorithmic problems of balanced loading and unloading. For unloading,
this means planning an optimal sequence for removing a given set of objects, one at a time; for
loading, this requires planning both position and order of the objects. 

The practical constraints of loading and unloading motivate a spectrum of relevant scenarios. 
As ships are symmetric around their main axis, we focus on one-dimensional settings, 
in which the objects correspond to intervals. Containers may be 
of uniform size, but stackable up to a certain limited height; because sliding objects on a moving ship are
major safety hazards, stability considerations may prohibit gaps between containers. On the other hand,
containers of extremely different size pose particularly problematic scenarios, which is why
we also provide results for sets of containers whose sizes are exponentially growing.

\subsection{Our Contributions}

Our results are as follows; throughout the paper, {\em items} are the objects
that need to be loaded (also sometimes called \emph{placed}) or unloaded, while {\em container} refers to the space
that accomodates them. Furthermore, we assume all objects to have unit density, i.e., their weights correspond to their lengths. In most cases, items correspond to geometric intervals.
	
	\begin{itemize}
		\item For unloading, we show that it is \NP-complete to compute
an optimal sequence. More formally, given a set of placed intervals $\{
I_1,\dots,I_n \}$, it is \NP-complete to compute an order $\langle I_{\pi(1)},
\dots,I_{\pi(n)} \rangle$, in which intervals are removed one at a time, such
that the maximal deviation of the gravity's center is minimized.  
		\item We
provide a corresponding polynomial-time 2.7-approximation algorithm. In
particular, we give an algorithm that computes an order of the input intervals
such that removing the intervals in that order results in a maximal deviation
which is no larger than $2.7$ times the maximal deviation induced by an optimal order.
		\item For loading, we give a polynomial-time algorithm
for the setting in which gaps are not allowed. In particular, given a set of
lengths values $\ell_1,\dots,\ell_n \in \mathbb{R}_{>0}$, we require a sequence
$\langle I_{\pi(1)}, \dots, I_{\pi(n)} \rangle$ of pairwise disjoint 
intervals with $|I_{\pi(i)}| = \ell_{\pi(i)}$ for $i = 1,\dots,n$ such that the
following holds: Placing the interval $I_{\pi(i)}$ in the $i$-th step results
in an $n$-stepped loading process such that the union of the loaded intervals
is connected for all points in time during the loading process. Among these
\emph{connected placements}, we compute one for which the
maximal deviation of the center of gravity is minimized.  
		\item We give a
polynomial-time algorithm for the case of stackable unit intervals. More
formally, given an input integer $\mu \geq 1$, in the context of the previous
variant, we relax the requirement that the union of the placed intervals has to
be connected and additionally allow that the placed intervals may be stacked up to
a height of $\mu$ in a stable manner, defined as follows. We say that layer $0$ is completely
\emph{covered}. An interval $I$ can be placed, i.e., covered, in layer $k \geq 1$ if the
interval $I$ is covered in all layers $0,\dots,k-1$ and if $I$ does not overlap
with another interval already placed in layer $k$.
		\item We give a polynomial-time algorithm for the case of
exponentially growing lengths. More formally, in the context of the previous
variant, we require that all intervals are placed in layer $1$ and assume that
the lengths of the input intervals' lengths are exponentially increasing, i.e.,
there is an $x \geq 2$ such that $x \cdot \ell_i = \ell_{i+1}$ holds for all $i
\in \{ 1,\dots,n-1 \}$.  \end{itemize}

\subsection{Related Work}

Previous work on cargo loading covers a wide range of specific aspects, constraints and objectives.
%
The general \textsc{Cargo Loading Problem} (CLP) asks for an optimal packing of
(possibly heterogeneous) rectangular boxes into a given bin, 
equivalent to the \textsc{Cutting Stock Problem}~\cite{Gilmore1965}. 
Most of the proposed methods are heuristics based on (mixed)
integer programming and have been studied both for 
heterogeneous and homogeneous items. Bischoff and 
Marriott~\cite{Bischoff1990comparative} show that the performances of some
heuristics may depend on the kind of cargo.

Amiouny et al.~\cite{AmiounyBVZ92} consider the problem of packing a set of
one-dimensional boxes of different weights and different lengths into a flat bin
(so they are not allowed to stack these boxes), in such a way that after
placing the last box, the center of gravity is as close as possible to a fixed
target point. They prove strong \NP-completeness by a
reduction from \textsc{3-Partition} and give a heuristic with a
guaranteed accuracy within $\ell_{max}/2$ of a given target point, where
$\ell_{max}$ is the largest box length. A similar heuristic is given by
Mathur~\cite{Mathur98}.

Gehring et al.~\cite{Gehring1990} consider the general CLP, in which (rectangular) items may be stacked and placed in any possible position. 
They construct non-intersecting
\emph{walls}, i.e. packings made from similar items for slices of the original container, to generate the overall packing. They also
show that this achieves a good final balancing of the loaded items.
Mongeau and Bes~\cite{Mongeau2003} consider a similar variant in which the
objective is to maximize the loaded weight. 
In addition, there may be other paramaters, 
e.g., each item may have a different priority \cite{Souffriau1992}. 
A mixed integer programming approach on this variant is
given by Vancroonenburg et al.~\cite{vancroonenburg2014automatic}. Limbourg et al.~\cite{LimbourgSL12} consider the CLP
based on the moment of inertia. Gehring and Bortfeldt~\cite{Gehring1997} give a
genetic algorithm for \emph{stable}
packings. Fasano~\cite{Fasano04} considers packing problems of
three-dimensional \emph{tetris}-like items in combination with balancing
constraints. His work is done within the context of the Automated Transfer
Vehicle, which was the European Space Agency's transportation system supporting
the International Space Station (ISS).

Another variant is to consider multiple {\em drops}, for which
loaded items have to be available at each drop-off point in such a way that a
rearrangement of the other items is not required; see e.g.~\cite{Bischoff1995}, \cite{Christensen2009}, and \cite{lurkin2015airline}. Davies and Bischoff~\cite{DaviesB99} propose an approach
that produces a high space utilization for even weight distribution. These
scenarios often occur in container loading for trucks, for which the objective is
to achieve an even weight distribution between the axles. For a state-of-the-art
survey of vehicle routing with different loading constraints and a spectrum of scenarios, 
see Pollaris et al.~\cite{pollaris2015vehicle}. 	

In the context of distributing cargo by sea, two
different kind of ships are distinguished: \emph{Ro-Ro} and \emph{Lo-Lo} ships. Ro-Ro (for
roll on--roll off) ships
carry wheeled cargo, such as cars and trucks, which are driven on and off the
ship. Some approaches and problem variants such as multiple drops, additional
loading, and optional cargo as well as routing and scheduling considering Ro-Ro
ships are considered by \O vsteb\o\ et
al.~\cite{ovstebo2011optimization,ovstebo2011routing}. On the other hand, Lo-Lo
(load on--load off) ships are cargo ships that are loaded and unloaded by cranes, so 
any feasible position can be directly reached from above.

While all of this work is related to our problem, it differs in not requiring the center of gravity to be under control for each step of the loading or unloading process. A problem in which a constraint is imposed at each step of a process is \textsc{Compact Vector Summation} (CVS), which asks for a permutation of a set of $k$-dimensional vectors in order to control their sum, keeping each partial sum within a bounded $k$-dimensional ball. See Sevastianov~\cite{sevast1994some,sevastianov1998nonstrict} for a summary of results in CVS and its application in job scheduling.

A different, but related 
stacking problem is considered by Paterson and Zwick~\cite{PatersonZ09}
and Paterson et al.~\cite{PatersonPTWZ09}, who consider maximizing the reach beyond the edge 
of a table by stacking $n$ identical, homogeneous,
frictionless bricks of length 1 without toppling over,
corresponding to keeping the center of gravity of subarrangements supported.

\section{Preliminaries}

	An \emph{item} is a unit interval $I:=[m-\frac{1}{2},m+\frac{1}{2}]$ with midpoint $m$. A set $\{ I_1,\dots,I_n \}$ of~$n$ items with midpoints $m_1,\dots,m_n$ is \emph{valid} if $m_i = m_j$ or $|m_i - m_j| \geq 1$ holds for all $i,j = 1,\dots,n$. The \emph{center of gravity} $\cog{I_1,\dots,I_n}$ of a valid set $\{ I_1,\dots,I_n \}$ of items is defined as $\frac{1}{n} \sum_{i = 1}^{n} m_i$.
	
	Given a valid set $\{ I_1,\dots,I_n \}$ of items, we seek orderings in which each item $I_j$ is removed or placed such that the maximal deviation for all points in time $j = 1, \dots, n$ is minimized. 
	More formally, for $j = 1,\dots,n$ and a permutation $\pi : j \mapsto \pi_j$, let $C_j := \cog{I_{\pi_j},\dots,I_{\pi_n}}$.
%
	The \textsc{Unloading Problem} (\UPr) seeks to minimize the maximal deviation during an unloading process of $I_1,\dots,I_n$. In particular, given an input set $\{I_1,\dots,I_n \}$ of items, we seek a permutation $\pi$ such that $\max_{i,j = 1,\dots,n} |C_i- C_j|$ is minimized. 
(Equivalently, the \UPr can be posed as a {\em loading} problem in which the positions of the items is given, and we seek an order of loading them.)
	
	In the \textsc{Loading Problem} (\LPr) we relax the constraint that
the positions of the considered items are part of the input. In particular, we seek an ordering
and a set of midpoints for the containers such that the containers are disjoint
and the maximal deviation for all points in time of the loading process is
minimized; see Section~\ref{sec:loading} for a formal definition.


\section{Unloading}\label{sec:unloading}

We show that the problem \UPr is
\NP-complete and give a polynomial-time $2.7$-approximation algorithm for \UPr.
We first show that there is a polynomial- time reduction
from the discrete version of \UPr, the \textsc{Discrete Unloading Problem}
(\dUPr), to \UPr; this leads to a proof that \UPr is \NP-complete, followed by a
$2.7$-approximation algorithm for \UPr.

In the \textsc{Discrete Unloading Problem} (\dUPr), we do not consider a set of
items, i.e., unit intervals, but a discrete set $X:= \{ x_1,\dots,x_n \}$ of
points. The center of gravity $\cog{X}$ of $X$ is defined as
$\frac{1}{n}\sum_{i=1}^n x_i$. For $j = 1,\dots,n$ and a permutation $\pi : j
\mapsto \pi_j$, let $C_j = \cog{x_{\pi_j},\dots,x_{\pi_n}}$. Again, we seek a
permutation such that $\max_{i,j=1,\dots,n}|C_i-C_j|$ is minimized.

	
	\begin{lemma}\label{lem:discreteUPrtointervalUPr}
		There is a polynomial-time reduction from \UPr to \dUPr.
	\end{lemma}
	\begin{proof}
		Consider the items $I_1,\dots, I_n$ with their midpoints $m_1,\dots,m_n$. We choose $x_i=m_i$ to get a discrete set of points.
		It is easy to see that the center of gravity does not change, i.e., after removing $k$ intervals, and for any permutation~$\pi$, we have 
		$\frac 1 {k+1} \sum_{i=n-k}^n x_{\pi_i} = \frac 1 {k+1} \sum_{i=n-k}^n m_{\pi_i}$. Because this holds for any $k$ and $\pi$ we conclude that
		an optimal solution to \dUPr is an optimal solution to \UPr.
	\end{proof}
	
	\begin{lemma}\label{red:intervalUPrtodiscreteUPr}
		There is a polynomial-time reduction from \dUPr to \UPr.
	\end{lemma}
	\begin{proof}
		Consider a set $X$ of points and the smallest distance $d$ among all pairs of points.
		We construct unit intervals $I_1,\dots,I_n$ with midpoints $m_i = \frac{x_i}d$.
		After removing $k$ intervals, and for any permutation $\pi$, the center of gravity is at position 
		$\frac 1 {k+1} \sum_{i=n-k}^n m_{\pi_i} = \frac 1 {d(k+1)}\sum_{i=n-k}^n x_{\pi_i}$, i.e.,
		the center of gravity is scaled by a factor of $d$.
		Because this holds for any $k$ and $\pi$ we conclude that
		an optimal solution to \UPr is an optimal solution to \dUPr.
	\end{proof}
	
	The combination of Lemma~\ref{lem:discreteUPrtointervalUPr} and Lemma~\ref{red:intervalUPrtodiscreteUPr} implies Corollary~\ref{cor:intervalUPrEQdiscreteUPr}.
	
	\begin{corollary}\label{cor:intervalUPrEQdiscreteUPr}
		\UPr and \dUPr are polynomial-time equivalent.
	\end{corollary}

\subsection{\NP-Completeness of the Discrete Case}

We can establish \NP-completeness of the discrete problem \dUPr.

\begin{theorem}\label{thm:discreteUPrNPcomplete}
                \dUPr is \NP-complete.
        \end{theorem}

\begin{proof}
Our reduction is from 3-{\sc Partition}.  An instance of 3-{\sc
  Partition} takes as input a multiset $Y=\{y_1,y_2,\ldots,y_{3m}\}$
of $3m$ positive integers and asks if it is possible to partition $Y$
into $m$ triples, $Y_1,Y_2,\ldots,Y_m$, such that the sum of the
integers in each triple $Y_i$ is exactly $B=(\sum_{j=1}^{3m} y_j)/m$.  The
problem 3-{\sc Partition} is known to be NP-complete, and it remains
NP-complete when one assumes that $B/4 < y_j < B/2$, for all
$j=1,2,\ldots,3m$. Given such an instance,
$Y=\{y_1,y_2,\ldots,y_{3m}\}$, we construct an instance of \dUPr
such that the set $Y$ has the desired partition into $m$ triplets if
and only if the instance of \dUPr has a {\em feasible ordering}, for which the
center of gravity is always within the interval $[0,B/M]$.

Our instance consists of the following set $X=\{x_1,x_2,\ldots\}$ of points:
\begin{description}
\item[(i)] $M$ points at the origin, 0, for a large integer $M>m$ to be specified below (we see that $M>4mB$ suffices). 
\item[(ii)] $3m$ points at the $3m$ integers $y_i\in Y$. 
\item[(iii)] $m$ points at $-B$.
\end{description}
In the following basic description, this instance $X$ of \dUPr is a multiset, i.e., it 
has repeated points: multiple points at 0,
at $-B$, and potentially at points of $Y$, so it forms a multiset. These
points can be perturbed to be distinct, and then the instance can be
rescaled so that the minimum distance between consecutive points is at
least~1.

We let
$\sigma_k=x_{\pi_1}+\cdots x_{\pi_k}$ denote the partial sum of the first $k$ points in the ordered sequence.
Our goal is to decide if there exists an ordering $\pi$ of the
$n=3m+m+M=|X|$ points of $X$ so that the centers of gravity,
$C_k=C(x_{\pi_1},\ldots,x_{\pi_k})$, remain within the interval
$[0,B/M]$, for all $k=1,2,\ldots,n$.
(For simplicity our description here is in terms of adding the points one by one; 
for the unloading order, time is reversed.)

The center of gravity of all $n=4m+M$ of the points of $X$ is 
$C_n=\frac{   M\cdot 0 + Bm + m\cdot (-B)}{4m+M} = 0.$

We claim that the centers of gravity $C_k$ will lie within the
interval $[0,B/M]$, for all $k=1,\ldots,n$, if and only if the points
$X$ are ordered as follows. First, all $M$ of the points of type~(i)
at the origin, then three points (in any order) of type~(ii) that sum
to exactly $B$, then one point at $-B$ of type~(iii), then three
points of type~(ii) that sum to exactly $B$, then one point of
type~(iii), etc.  This means that the centers of gravity $C_k$ will
lie within the interval $[0,B/M]$, for all $k=1,\ldots,n$, if and only
if there is a partition of $Y$ into triples, each summing to exactly
$B$.

The ``if'' direction of the claim is straightforward: If $Y$ has a
partition into triples, each summing to exactly $B$, then the specified
ordering of the points $X$ ($M$ points at 0, followed by a triple that
sums to exactly $B$, followed by a point at $-B$, followed by another
triple that sums to exactly $B$, etc) achieves the desired containment
of the center of gravity, $C_k$, in the interval $[0,B/M]$, for each
$k$.

For the ``only if'' claim, first note that if the sequence does not
begin with $M$ elements all at the origin, having instead a non-zero
element at position $k<M$, then the center of gravity $C_k$ is either
$-B/k < 0$ or is at least $1/k>1/M$, because each of the type~(ii)
points is a positive integer from the set $Y$.  Thus, a feasible
sequence $\pi$ (that maintains the center of gravity within $[0,B/M]$)
must begin with $M$ elements of type~(i), namely points at the origin.
If the next point is of type~(iii), at $-B$, then the center of
gravity becomes negative, and falls outside of $[0,B/M]$; thus, the
$(M+1)$st element must be a positive integer.  The partial sum
$\sigma_k=x_{\pi_1}+\cdots x_{\pi_k}$ of the first $k$ points in the
sequence must never become negative (or else the center of gravity, $\sigma_k/k$,
will fall below the target interval $[0,B/M]$) and must never become
greater than $B$, or else the center of gravity will fall above the
target interval $[0,B/M]$, as we now argue. If $\sigma_k>B$, then
$\sigma_k \geq B+1$, making the center of gravity at least
$\frac{B+1}{k}$.  Now if we pick $M$ to be large enough (it suffices
to pick $M>4mB$), then $\frac{B+1}{k} > \frac{B}{M}$, for all
$k=1,2,\ldots,M+4m=n$, as claimed. Thus, in order for the center of
gravity $C_k$ to remain in $[0,B/M]$ for all $k$, the partial sum must
never get above $B$ and must never become negative. The only way this
can be accomplished is to sequence the points of types~(ii) and (iii)
as claimed, into triples of points of $Y$ that sum to exactly $B$,
followed by a point $-B$, followed by another triple of points of $Y$
that sum to exactly $B$, etc.  We conclude that if the instance $X$
has a solution, keeping the center of gravity within the target
interval $[0,B/M]$, then the input instance of 3-{\sc Partition}, $Y$,
has a solution.
	\end{proof}
	
Because of the polynomial-time equivalance of \dUPr and \UPr, we conclude the following.
	
	\begin{corollary}\label{thm:discreteUPrNPcomplete}
		\UPr is \NP-complete.
	\end{corollary}



\subsection{Lower Bounds and an Approximation Algorithm}

When unloading a set of items, their positions are fixed, so (after reversing time)
unloading is equivalent to a loading problem with predetermined positions.
For easier and uniform notation throughout the paper, we use this latter description.

In order to develop and prove an approximation algorithm for \dUPr, we
begin by examining lower bounds on the span, $R-L$, of a minimal
interval, $[L,R]$, containing the centers of gravity at all stages in
an optimal solution. 

Without loss of generality, we assume that the input points $x_i$ sum
to 0 (i.e., $\sum_i x_i = 0 $), so that the center of gravity, $C_n$, of all
$n$ input points is at the origin. We let $R = \max_i C_i$ and $L = \min_i C_i$.  
Our first simple lemma leads to a first (fairly weak) bound on the span.

\begin{lemma} \label{lem:naiveBound}
Let $(x_1,x_2,x_3,\ldots)$ be any sequence of real numbers, with $\sum_i x_i = 0$. 
Let $C_j=(\sum_{i=1}^j x_i)/j$ be the center of
gravity of the first $j$ numbers, and let $R=\max_i C_i$ and $L=\min_i C_i$.
Then, $|R - L| \geq \frac{|x_i|}{i}$, for all $i=1,2,\ldots$.
\end{lemma}

\begin{proof}
We see that 
$$C_i = \frac{ x_i + \sum_{j=1}^{i-1} x_j }{i} = \frac{i-1}{i} C_{i-1} + \frac{x_i }{i};$$
thus, 
$$C_i - C_{i-1} = \frac{x_i - C_{i - 1} }{i}.$$
We now distinguish two cases. If $x_i$ and $C_{i-1}$ have opposite signs, then 
we get 
$$  |C_i - C_{i-1}| = \frac{|x_i - C_{i - 1}| }{ i } = \frac{|x_i| + |C_{i-1}|}{i} \geq |\frac{x_i}{i}|,$$
implying that
$$  |R - L| \geq |C_i - C_{i-1}| \geq \frac{|x_i|}{i},$$
because the interval $[L,R]$ must be large enough to contain any amount of change, $|C_i-C_{i-1}|$, to the center of gravity, at any step $i$.
On the other hand, if $x_i$ and $C_{i-1}$ have the same signs (i.e., $\textrm{sign}(C_{i-1} x_i) = 1$), 
then $x_i$, $C_i$, and $C_{i-1}$ all have the same signs, and we get
$$  |C_i| =  \frac{i-1}{i} |C_{i-1}| + \frac{|x_i|}{i} \geq \frac{|x_i|}{i},$$
implying that 
$$  |R - L| \geq |C_i| \geq \frac{|x_i|}{i},$$
because the interval $[L,R]$ must be large enough to contain both $C_i$ and the overall center of gravity, 0, for any $i$.
Thus, in all cases, $|R-L|\geq \frac{|x_i|}{i}$.
\end{proof}

\begin{corollary} \label{cor:naiveBound}
For any valid solution to \dUPr, the minimal interval $[L,R]$ containing the center of 
gravity at every stage must have length $|R-L| \geq \frac{|u_i|}{i}$ where $u_i$ is the 
input point with the $i$-th smallest magnitude.
\end{corollary}

We note that the naive lower bound given by Corollary~\ref{cor:naiveBound}
can be far from tight: 
Consider the sequence 
$1, 2, 3, 4, 5, 6, 7, -7, -7, -7, -7$.  In the optimal order, the first 
$-7$ is placed fourth, after $2, 1, 3$.  The optimal third and fourth centers, 
$\{2, -\frac{1}{4}\}$ are the largest magnitude positive and negative 
centers seen, and show a span $2.25$ times greater than the naive bound of $1$.
By placing the first $-7$ in the third position, $R \geq \frac32$, and 
$L \leq -\frac43$.  By placing it fifth, $R \geq \frac52$.  Our observation
was that failing to place our first $-7$ if the cumulative sum is $> 7$ would 
needlessly increase the span.

This generalizes to the sequence $(x_1=1,x_2=2,\ldots, x_{k-1}=k-1, x_{k}=-k, 
x_{k+1}=-(k+1),\ldots,x_N)$, with an appropriate $x_N$ to make 
$\sum x_i = 0$.  If we place positive weights in increasing order until 
$ c_{l} \geq \frac{k}{l}$,  placing $-k$ instead of a positive at position $l$ 
would decrease the center of gravity well below $\frac{k}{l}$. The first 
negative should be placed when $ \min_l \frac{l^2 - l}{2} \geq k$, which 
is when $l \approx \sqrt{2k}$.  
In this example, our optimal center of gravity span is at least 
$\frac{k}{l} \approx \sqrt{\frac{k}{2}}$, not the 1 from the 
naive bound of Corollary~\ref{cor:naiveBound}.

We now describe our heuristic, ${\cal H}$, which leads to a provable
approximation algorithm.  It is convenient to relabel and reindex the
input points as follows. Let $(P_1, P_2, \ldots)$ denote the positive
input points, ordered (and indexed) by increasing value. Similarly,
let $(N_1, N_2, \ldots)$ denote the negative input points, orders (and
indexed) by increasing magnitude $|N_i|$ (i.e., ordered by decreasing
value).

The heuristic ${\cal H}$ orders the input points as follows. The first
point is simply the one closest to the origin (i.e., of smallest
absolute value).  Then, at each step of the algorithm, we select the
next point in the order by examining three numbers: the partial sum,
$S$, of all points placed in the sequence so far, the smallest
magnitude point, $\alpha$, not yet placed that has the same sign as
$S$, and the smallest magnitude point, $\beta$, not yet placed that
has the opposite sign of $S$. If $S+\alpha+\beta$ is of the same sign
as $S$, then we place $\beta$ next in the sequence; otherwise, if
$S+\alpha+\beta$ has the opposite sign as $S$, then we place $\alpha$
next in the sequence. The intuition is that we seek to avoid the
partial sum from drifting in one direction; we switch to the opposite
sign sequence of input points in order to control the drift, when it
becomes expedient to do so, measured by comparing the sign of $S$ with
the sign of $S+\alpha+\beta$, where $\alpha$ and $\beta$ are the
smallest magnitude points available in each of the two directions.
We call the resulting ordering the ${\cal H}$-permutation.
The ${\cal H}$-permutation puts the $j$-th largest positive point, $P_j$, 
in position $\pi^+_j$ in the order,
and puts the $j$-th largest in magnitude negative point, $N_j$, in position 
$\pi^-_j$ in the order, where
$$
\pi^+_j = j + \max_k \{ k \; :\; \sum_{i=1}^k |N_i| \leq \sum_{l=1}^j P_l \} \mbox{ \ and \ } 
\pi^-_j = j + \max_k \{ k \; :\; \sum_{i=1}^k P_i < \sum_{l=1}^j |N_l| \}. 
$$

We obtain an improved lower bound based on our heuristic, ${\cal H}$,
which orders the input points according to the ${\cal H}$-permutation.

\begin{lemma} \label{lem:tentpoleBound}
A lower bound on the optimal span of \dUPr is given by
$|R-L| \geq \frac{P_{i}}{\pi^+_i}$ and $|R-L| \geq \frac{|N_{i}|}{\pi^-_i}$.
\end{lemma}

To prove the lemma, we begin with a claim.


\begin{claim} \label{claim:increasingOrder} 
For any input set to the discrete unloading problem, where $s_i$ are all terms with the same sign sorted by magnitude, 
a permutation $\pi$ that minimizes the maximum value of the ratio $\frac{|s_i|}{\pi_i}$ must satisfy $\pi_k < \pi_i$, for all $k<i$.
\end{claim}

\begin{proof}
By contradiction, assume that the minimizing permutation $\pi$ has the
maximum value of the ratio $\frac{|s_i|}{\pi_i}$ occur at an $i$ for
which there exists a $k<i$ for which $\pi_i \leq \pi_k$, which means
that $\pi_i < \pi_k$ (because $\pi_i$ cannot equal $\pi_k$ for a
permutation $\pi$, and $k\neq i$).

Because the terms $s_i$ are indexed in order sorted by magnitude, $|s_k| \leq |s_i|$.
Exchanging the order of $s_i$ and $s_k$ in the
permutation would lead to two new ratios in our sequence:
$\frac{|s_i|}{\pi_k}$ and $\frac{|s_k|}{\pi_i}$.
Because $\pi_k > \pi_i$, we get $\frac{|s_i|}{\pi_k} < \frac{|s_i|}{\pi_i}$.
Because $|s_k| \leq |s_i|$, we get $\frac{|s_k|}{\pi_i} \leq \frac{|s_i|}{\pi_i}$.
Because these new ratios are smaller than $\frac{|s_i|}{\pi_i}$, we get a contradiction to the
fact that $\pi$ minimizes the maximum ratio.
\end{proof}

The following claim is an immediate consequence of Lemma~\ref{lem:naiveBound}.

\begin{claim} \label{claim:tentBeforeBad}
For the $i$ maximizing $\frac{P_{i}}{\pi^+_i}$, any ordering placing this element earlier
than $\pi^+_i$ in the sequence has a span $|R-L| > \frac{P_i}{\pi^+_i}$.
Similarly, for the $i$ maximizing $\frac{|N_{i}|}{\pi^+_i}$, any ordering placing this element earlier
than $\pi^-_i$ in the sequence has a span $|R-L| > \frac{|N_i|}{\pi^-_i}$.
\end{claim}

On the other hand, the following holds.

\begin{claim} \label{claim:tentAfterBad}
For the $i$ maximizing $\frac{P_{i}}{\pi^+_i}$, any ordering placing this element later than
$\pi^+_i$ in the sequence has a span $|R-L| > \frac{P_i}{\pi^+_i}$.
A similar statement holds for $\frac{|N_{i}|}{\pi^-_i}$.
\end{claim}

\begin{proof}
  The proof is by contradiction.  
The index into the ${\cal H}$ permutation maximizing the ratio
$\frac{|x_k|}{k}$ is $i$.  We assume (wlog) $x_i = P_J > 0$,
and we let $K = i - J$.

If $P_J$ is not placed in position $i$, we suppose another element,
$x$, can be placed in its stead and results in a span that is less
than $\frac{P_J}{i}$.

When placing any positive $x > P_J$ in the initial $i$
position, the lowest possible observed span from Lemma
\ref{lem:naiveBound} is at least $\frac{x}{i} > \frac{P_J}{i}$, which
would contradict our assumption.  Similarly, all positive points placed
before or at position $i$ must be less than or equal to $P_J$.

All permutations of these $J-1$ positive elements and the first $K+1$ negative elements
have a large negative center of gravity at position $i$.
From $ K = \max_k \{ k \; :\; \sum_{i=1}^k |N_i| \leq \sum_{l=1}^J P_l \}$,
we get $\sum_{i=1}^{K+1} |N_i| \geq \sum_{l=1}^{J} P_l$,
and hence $\sum_{i=1}^{K+1} N_i  + \sum_{l=1}^{J} P_l  \leq 0$,
implying $\sum_{i=1}^{K+1} N_i  + \sum_{l=1}^{J-1} P_l \leq -P_J$.
Therefore, the maximizing value satisfies
$$
    |c^*| = \frac{| \sum_{i=1}^{K+1} N_i  + \sum_{l=1}^{J-1} P_l |}{i} \geq \frac{P_J}{i} 
$$
Because the center of gravity is at a location greater than the ${\cal H}$-bound,
and $R \geq 0 \geq L$, this span is also greater than the ${\cal H}$-bound and we
can neither place an element greater than $P_J$ nor one less than $P_J$ in place of $P_J$
while lowering the span beneath the ${\cal H}$-bound.  
\end{proof}

\begin{theorem} \label{thm:tentpoleBound}
The ${\cal H}$-permutation minimizes the maximum (over $i$) value of
the ratio $\frac{|x_i|}{\pi_i}$, and thus yields a lower bound on $|R
- L|$.
\end{theorem}

For the worst-case ratio, we get the following.

\begin{theorem} \label{thm:tentpoleHeuristicBound}
The ${\cal H}$ heuristic yields an ordering having span $R-L$ at most 2.7 times larger than the ${\cal H}$-lower bound.
\end{theorem}

\begin{proof}
Before separating the input points into the sorted $P$ and $N$ lists,
we normalize their values so that the maximum value of the ratio
$\frac{|x_i|}{i}$ is 1.  This implies that $|x_i | \leq i$, for all
$i$.

When using the ${\cal H}$-permutation, whenever we place an element of
opposite sign from the current center of gravity, $C_i$, we know that the partial sum
$S_i$ and center of gravity $C_i$ obey $|S_i|  \leq |x_{i+1}|$. Using normalization, we have
  $|x_{i+1}| \leq i + 1$, hence 
  $|C_i| = |\frac{S_i}{i}| \leq \frac{i+1}{i}$.

When the center of gravity reaches its leftmost extent, we cannot
place another negative element, because the next largest negative
element would push the center of gravity further to the left.  A
similar statement holds for the rightmost extent and positive
elements.  This means that if the center of gravity first reaches $L$ at step $a$
and first reaches $R$ at step $b$, then
  $L \geq \frac{-(a+1)}{a}$ and $R \leq \frac{b+1}{b}$ (*), so $R - L \leq 2 + \frac{1}{a} + \frac{1}{b}$ (**).

The final step is to argue that the right-hand side is bounded by $2.7$.

%





%
 %
  
Let us assume $a < b$ and consider the ratio $\alpha(a,b)$ of the span we obtain from
${\cal H}$ to the ${\cal H}$-bound. We consider small values of $a$ and $b$ and can provide the following
bounds. 

$\mathbf{\alpha(1, b) \leq 2.5}$. This holds as follows.
We have $L = \frac{N_1}{1} \geq \frac{-1}{1}$ by normalization of masses.
As $b \geq 2$, it follows from 
Equation~(*) that 
$R \leq \frac{3}{2}$.  Thus, $R - L \leq 2 \frac12$. 

$\mathbf{\alpha(2, 3) \leq 1.5} $. This follows from 
considering the first terms: We observe that $L = \frac{N_1 + N_2}{2}$ 
and $R =  \frac{N_1 + N_2 + P_1}{3}$.
  Moreover, $P_1 \leq 3$ by normalization of masses and $|N_1 + N_2| \leq P_1$ from the 
${\cal H}$-ordering on $P_1$. Therefore, $R - L = \frac{2 P_1 - N_1 - N_2 }{6} \leq \frac{3 P_1}{6} \leq \frac{3}{2}$. 

$\mathbf{\alpha(2, 4) \leq 2.5}$. Again observe that 
  $L = \frac{N_1 + N_2}{2}$, as well as $R =  \frac{N_1 + N_2 + P_1 + P_2}{4}$, so 
  $R - L = \frac{P_1 + P_2 - N_1 - N_2}{4}$. By ${\cal H}$-ordering, we get $|N_1 + N_2| \leq P_1$, so  
  $R - L \leq \frac{2P_1 + P_2}{4}$. Again by normalization of masses, we have
 $P_1 \leq 3$ and $P_2 \leq 4$, implying 
  $R - L \leq \frac{10}{4} = 2.5$

Therefore, we only have to consider \textbf{$\alpha(2, b \geq 5) \leq 2.7$}, which follows
from Equation~(**).
\end{proof}
        
\begin{corollary}\label{thm:constFactApprdiscreteUPr}
  There is a polynomial-time $2.7$-approximation algorithm for \UPr.
\end{corollary}

\section{Loading}\label{sec:loading}

We now consider loading problems, for which we require some additional definitions:
The positions of the objects are part of the optimization, and, for some loading
variants, the items may have different lengths. 
Consider the following more general definitions.

	An \emph{item} is given by a real number $\ell$. By assigning a
\emph{position} $m \in \mathbb{R}$ to an item, we obtain an interval $I$ with
length $\ell$ and midpoint $m$. For $n \geq 1$, we consider a set $\{ \ell_1,\dots,\ell_n \}$ of $n$ items and assume $\ell_1 \geq \dots \geq \ell_n$. Furthermore, $\{ \ell_1,\dots,\ell_n \}$ is \emph{uniform} if $\ell:= \ell_1 =
... = \ell_n$.

	A \emph{state} is a set $\{ (I_1,h_1),\dots,(I_n,h_n) \}$ of pairs,
each one consisting of an interval $I_j$ and an integer $h_j \geq 1$, the
\emph{layer} in which $I_j$ lies. A state satisfies the following: (1) Two
different intervals that lie in the same layer do not overlap and (2) for $j =
2,\dots,n$, an interval in layer $j$ is a subset of the union of the intervals
in layer $j-1$.
	
	A state $\{ (I_1,h_1),\dots,(I_n,h_n) \}$ is \emph{plane} if all intervals lie in the first layer.
	
	To simplify the following notations, we let $m_j$ denote the midpoint of the interval $I_j$, for $j = \{ 1,\dots,n \}$. The \emph{center of gravity} $\cog{s}$ of a state $s = \{ (I_1,h_1),\dots,(I_n,h_n) \}$ is defined as $\frac{1}{M}\sum_{j = 1}^n \ell_j m_j$, where $M$ is defined as $\sum_{j=1}^n \ell_j$.
	
	A \emph{placement} $p$ of an $n$-system $S$ is a sequence $\langle I_1,\dots,I_n \rangle$ such that $\{ (I_1,h_1)$, $\dots, (I_j,h_j)\}$ is a state, the \emph{$j$-th state $s_j$}, for each $j = 1,\dots,n$. The $0$-th state $s_0$ is defined as $\emptyset$ and its center of gravity $\cog{s_0}$ is defined as $0$. 
	
	\begin{definition}
	The \textsc{loading problem} (\LPr) is defined as follows: Given a set of $n$ items, determine a placement $p$ such that the $n+1$ centers of gravity of the $n+1$ states of $p$ lie close to $0$. In particular, the \emph{deviation} $\Delta(p)$ of a placement $p$ is defined as $\max_{j = 0,\dots,n}|\cog{s_j}|$. We seek a placement of $S$ with minimal deviation among all possible placements for $S$.
	
	We say that \emph{stacking is not allowed} if we require that all intervals are placed in layer $1$. Otherwise, we say that \emph{stacking is allowed}. For a given integer $\mu \geq 1$ we say that $\mu$ is the \emph{maximum stackable height} if we require that all used layers are no larger than $\mu$.
	\end{definition}

	Note that in the loading case, minimizing the deviation is equivalent
to minimizing the diameter, i.e., minimizing the maximal distance between the
smallest and largest extent of the centers.

	\subsection{Optimally Loading Unit Items With Stacking}
	
		Now we consider the case of unit items for which stacking is allowed. We give an algorithm that optimally loads a set of unit items with stacking.
		
\begin{theorem}\label{thm:stacking-unit-items}
			There is a polynomial-time algorithm for loading a set of unit items so that the deviation of the center of gravity is in $[0,\frac{1}{1+\mu}]$, where $\mu$ is the maximum stackable height.
		\end{theorem}
		\begin{proof}

Let $m_i$ be the midpoint of item $\ell_i$. Because we are allowed to stack
items up to height $\mu$, the strategy is the following: set $m_1 = m_2 =
\dots = m_{\mu} = \frac{1}{1+\mu}$, i.e., the first $\mu$ items are placed at
the very same position. Call these first $\mu$ items the {\em starting stack}
$\mathcal{S}_0$. Subsequently, we place the following items on alternating sides
of $\mathcal{S}_0$, i.e., the item $\ell_{\mu + 1}$ is placed as close as possible on the left side of $\mathcal{S}_0$, $\ell_{\mu + 2}$ is placed as close as
possible on the right side, $\ell_{\mu + 3}$ is placed on top of $\ell_{\mu + 1}$ (if we did not already reach the maximum stackable height of $\mu$), or next to $\ell_{\mu + 1}$ (if $\ell_{\mu + 1}$ is on the $\mu$-th layer), etc.

After each placement of $\ell_i, 1 \leq i \leq \mu$, we have $\cog{\ell_{i}} =
\frac{1}{1+\mu}$. After two more placed items, the center of gravity is again
at $\frac{1}{1+\mu}$, because these items neutralize each other. Thus, the critical
part is a placement on the left side of $\mathcal{S}_0$. We proceed to show that
after placing an item on the left side, the center of gravity is at position at
least $0$.

The midpoint $m_{\mu + 1}$ of the item $\ell_{\mu + 1}$ is
$\frac{-\mu}{1+\mu}$, thus $\cog{\ell_{\mu + 1}} = \frac{\mu}{1+\mu} -
\frac{\mu}{1+\mu} = 0$. Now assume that we have already placed $c =
(2k+1)\cdot \mu + \zeta$ items, where $\zeta < 2\mu$ and odd, i.e., we have
already placed the starting stack $\mathcal{S}_0$ and $k$ additional stacks of
height $\mu$ on each side of $\mathcal{S}_0$. Let $z := (2k+1)\cdot \mu$. Then
the center of gravity is at position $\cog{c}$, where 
\begin{align*}
	\cog{c} = & \frac{z\cdot\frac{1}{1+\mu} + \sum \limits_{i=z+1}^{z+\zeta} m_i}{z + \zeta} = \frac{(z+\zeta-1)\cdot\frac{1}{1+\mu} + \frac{-k\mu - k - \mu}{1+\mu}}{z+\zeta} = \frac{k\mu+\zeta-1-k}{(1+\mu)\cdot(z+\zeta)}\\
	= & \frac{k(\mu-1) + \zeta - 1}{(1+\mu)\cdot(z+\zeta)} \geq \frac{\zeta-1}{(1+\mu)\cdot(z+\zeta)} \geq \frac{0}{(1+\mu)\cdot(z+\zeta)} \geq 0.
\end{align*}
\end{proof}

In the following we show that there is no strategy that can guarantee a smaller
deviation of the center of gravity than the strategy described in the last theorem.

\begin{theorem}
	The strategy given in Theorem~\ref{thm:stacking-unit-items} is optimal for $n > \mu$, i.e., there is no strategy such that the center of gravity deviates in $[0,\frac{1}{1+\mu})$.	
\end{theorem}
\begin{proof}
	Because $n > \mu$, we must use at least two stacks. Now assume that we
first place $k$ items on one stack $\mathcal{S}_0$, before we start another
one. Without loss of generality, we place this first $k$ items at position
$\frac{1}{1+\mu}-\varepsilon$. We proceed to show that for any $\varepsilon > 0$, we
need $k$ to be at least $\mu+1$, to get the new center of gravity to position
$>-\varepsilon$ and therefore a smaller deviation as the strategy in
Theorem~\ref{thm:stacking-unit-items}.
	
	If we place the item $\ell_{k+1}$ on the right side of
$\mathcal{S}_0$, the new center of gravity gets to a position larger than
$\frac{1}{1+\mu}-\varepsilon$, a contradiction. Thus, it must be placed on the left of
$\mathcal{S}_0$. The position of this item has to be
$-\frac{\mu}{1+\mu}-\varepsilon$.  This yields the new center of gravity of
$(k\cdot (\frac 1 {1+\mu}-\varepsilon) - \frac \mu{1+\mu} - \varepsilon) /
k+1$.
	This center of gravity must be located to the right of $-\varepsilon$. Thus, we have
	\begin{align*}
&k\cdot (\frac 1 {1+\mu}-\varepsilon) - \frac \mu{1+\mu} - \varepsilon + (k+1)\cdot\varepsilon > 0	
		\quad	\Leftrightarrow\quad k - \mu> 0\quad \Leftrightarrow\quad k > \mu
	\end{align*}
	Because we cannot stack $\mu+1$ items, we cannot have any strategy achieving a deviation of $[0,\frac{1}{1+\mu}-\varepsilon]$. We conclude that our strategy given in Theorem~\ref{thm:stacking-unit-items} must be optimal.
\end{proof}

\begin{corollary}
	With the given strategy for a uniform system where each item has length $\ell$, the center of gravity
deviates in $[0, \frac \ell {1+\mu}]$, which is optimal.  
\end{corollary}

	\subsection{Optimally Loading Without Stacking but With Minimal Space}
	
		Assume that the height of the ship to be loaded does not allow
stacking items. This makes it necessary to ensure that the space
consumption of the packing is minimal.  We restrict ourselves to plane placements such that
each state is connected. For simplicity, we assume w.l.o.g.\ that $\ell_1 \geq
\dots \geq \ell_n$ holds. First we argue that $\Delta(p) \geq \frac{\ell_2}{4}$
holds for an arbitrary connected plane placement $p$ of~$S$. Subsequently we
give an algorithm that realizes this lower bound.
		
	A fundamental key for this subcase is that the center of gravity of a connected plane state is the midpoint of the induced overall interval.
	
	\begin{observation}\label{obs:interval}
		Let $s$ be a plane state such that the union of the corresponding intervals is an interval $[a,b]\subset \mathbb{R}$. Then $\cog{s} =  \frac{a+b}{2}$.
	\end{observation}
	
	\begin{lemma}\label{lem:lowerboundconnected}
		For each plane placement $p$ of $S$, we have $\Delta(p) \geq \frac{\ell_2}{4}$.
	\end{lemma}
	\begin{proof}
		Let $p$ be an arbitrary plane placement of $S = \langle
(I_1,1), \dots,(I_n,1) \rangle$, let $\langle s_0, s_1, \dots, s_n \rangle$ be the
sequence of states that are induced by $p$, and let $i,j \in \{ 1,\dots,n \}$ be such
that $I_i = |\ell_{1}|$ and $I_j = |\ell_{1}|$ hold.
Observation~\ref{obs:interval} implies that $|\cog{s_{i-1}} - \cog{s_i}| =
\frac{\ell_1}{4} \geq \frac{\ell_2}{4}$ and $|\cog{s_{j-1}} - \cog{s_j}| =
\frac{\ell_2}{4}$. Let $m_i$ and $m_j$ be the midpoints of $I_i$ and $I_j$. As
the intervals $I_i$ and $I_j$ do not overlap, we conclude that $|m_i| \geq
\frac{\ell_2}{2}$ or $|m_j| \geq \frac{\ell_2}{2}$ holds. W.l.o.g.\ assume that
$|m_i| \geq \frac{\ell_2}{2}$ holds. This implies that $|\cog{s_{i-1}}| \geq
\frac{\ell_2}{4}$ or $|\cog{s_{i}}| \geq \frac{\ell_2}{4}$ holds. In both
cases, we obtain $\Delta(p) \geq \frac{\ell_2}{4}$, concluding the
proof.  \end{proof}
	
	\begin{lemma}\label{lem:upperboundconnected}
We can compute a placement $p$ of $S$ such
that $\Delta(p) \leq \frac{\ell_2}{4}$.
	\end{lemma}

	\begin{proof} The main idea is as follows. We remember
$\ell_1\geq \dots \geq \ell_n$ and place the items in that order. In
particular, we choose the positions of the items such that $\cog{s_1} := -
\frac{\ell_2}{4}$ and $\cog{s_2} := \frac{\ell_2}{4}$. The remaining intervals are
placed alternating, adjacent to the left and to the right side of the
previously placed intervals.

		
		In order to show that $\cog{s_i} \in
[-\frac{\ell_2}{4},\frac{\ell_2}{4}]$ holds for all $i \in \{ 0,\dots,n \}$, we
prove by induction that $\cog{s_i} \in [\cog{s_{i-2}},\cog{s_{i-1}}]$ holds for
all odd $i \geq 3$ and $\cog{s_i} \in [\cog{s_{i-1}},\cog{s_{i-2}}]$ for all
even $i \geq 4$. As Observation~\ref{obs:interval} implies $\cog{s_1} =
-\frac{\ell_2}{4}$ and $\cog{s_2} = \frac{\ell_2}{4}$, this concludes the
proof.
		
		Let $i \geq 3$ be odd. We have $|\cog{s_{i-2}} - \cog{s_{i-1}}|
= \frac{\ell_{i-1}}{2}$. This is lower bounded by $\frac{\ell_{i}}{2}$
because~$\ell_i \leq \ell_{i-1}$. Furthermore, we know that $|\cog{s_{i-1}} -
\cog{s_i}| = \frac{\ell_{i}}{2}$. This implies $\cog{s_{i}} \in
[\cog{s_{i-2}},\cog{s_{i-1}}]$.  The argument for the case of even $i \geq 4$ is analogous. 
	\end{proof}
	
	The combination of Lemma~\ref{lem:lowerboundconnected} and Lemma~\ref{lem:upperboundconnected} implies that our approach for connected placements is optimal.
	
		\begin{corollary}\label{cor:woStackingAndMinimalSpace}
			Given an arbitrary system, there is a polynomial-time
algorithm for optimally loading a general set of items without stacking and
under the constraint of minimal space consumption for all intermediate stages.
\end{corollary}
	
\subsection{Optimally Loading Exponentially Growing Items}\label{sec:heterogenous}
	
		Similar to the previous section, we consider plane placements. 
Now we consider the case in which the items have exponentially rising lengths.
This case highlights the challenges of uneven lengths, in particular when the
sizes are growing very rapidly; without special care, this can easily lead to strong deviation
during the loading process. We show how the deviation can be minimized.
	
		\begin{theorem}\label{thm:exponential}
			There is a polynomial-time algorithm for optimally
loading a set of items with lengths growing exponentially by a factor $x \geq
2$ in increasing order w.r.t.\ to their lengths. 
		\end{theorem}

In the following, we describe a proof of 
Theorem~\ref{thm:exponential}. In particular, we
consider a system $S= \{ \ell_1,\dots,\ell_n \}$ for $n \geq 4$, 
i.e., there is an $x \geq 2$ such that $\ell_{i+1} = x \ell_i$ for all $i \in \{
1,\dots,n-1 \}$. 

First we describe the general approach of the proof and then give the details of the single steps (Lemma~\ref{lem:tightbound}, Lemma~\ref{lem:strategy1optimal}, and Lemma~\ref{lem:validplacementheterogenous}) of the general approach in Sections~\ref{app:lowerbound},~\ref{sec:prooflemstrategy1optimal}, and~\ref{sec:prooflemvalidplacementheterogenous}
		
		We establish a lower bound~$\tau$ for the deviation of any
placement of~$S$. The high-level idea of our approach is to place the largest interval
first and slightly shifted beside~$0$, such that the deviation that is caused
by the second largest interval is balanced; see
Figure~\ref{fig:captureTheSubpath}(a)+(b) for an illustration. 
Before placing the second largest interval, we place all remaining intervals in 
increasing order, alternating to
the right side and the left side of the largest interval, such that the centers of the
states alternate between $-\tau$ and~$\tau$. 


	In order to make the largest and second largest interval balance each
other, we guarantee $\cog{s_1} = - \tau$ and $\cog{s_n} = \tau$ by considering 
$\cog{s_1} = - \tau$ and $\cog{s_n}$ for even $n$,
and 
placing the second and the third interval on the same side such that $-\tau = \cog{s_1} <
\cog{s_2} <\cog{s_3} = \tau$ holds for odd $n$.

	\begin{definition}\label{def:tau}
		$\tau := \tau(S) := \frac{\ell' + \ell''}{4\sum_{j=1}^{n} \ell_j} \ell''$.
	\end{definition}

Further proof steps are according to the following sequence of lemmas.	
		
	\begin{lemma}\label{lem:tightbound}
		We have $\Delta(p) \geq \tau$ for each placement $p= \langle (I_1,h_1),\dots,(I_n,h_n) \rangle$ of $S$.
	\end{lemma}

The following lemma guarantees that the deviation of $p$ is equal to $\tau$. 

\begin{lemma}\label{lem:strategy1optimal}
	We have $\Delta(p) = \tau$ for the placement $p$ computed by the above algorithm.
\end{lemma}

Finally, we prove by Lemma~\ref{lem:validplacementheterogenous} that the
intervals in $p$ do not overlap.

\begin{lemma}\label{lem:validplacementheterogenous}
	The intervals as computed by the algorithm from above are pairwise disjoint.
\end{lemma}


\begin{proof}[Proof of Theorem~\ref{thm:exponential}]
	The combination of the Lemmas~\ref{lem:tightbound},~\ref{lem:strategy1optimal}, and~\ref{lem:validplacementheterogenous} guarantees that the above algorithm computes a placement with optimal deviation.
	
	The runtime is dominated by the time needed to compute the order of $\ell_1,\dots,\ell_n$, which takes time $\mathcal{O}(n \log n)$.
\end{proof}

\subsection{Proof of Lemma~\ref{lem:tightbound}}\label{app:lowerbound}

\statement{Lemma}{lem:tightbound}
               {\em
               		We have $\Delta(p) \geq \tau$ for a placement $p= \langle (I_1,h_1),\dots,(I_n,h_n) \rangle$ of $S$.
}

\medskip
In Figure~\ref{fig:captureTheSubpath}(a) we illustrate an optimal placement $p'$ for a $4$-system $S'$ and in Figure~\ref{fig:captureTheSubpath}(b) an placement $p''$ of a $2$-system $S''$ with  $S'=\{ \ell_1, \ell_2\}\subset S''= \{ \ell_1,\ell_2,\ell_3,\ell_4 \}$. 
	
\begin{figure}[ht]
  \begin{center}
    \begin{tabular}{ccccc}


\tikzset{
    punkt/.pic={
        \draw[x=\pgflinewidth,y=\pgflinewidth]
        (-2,-2)--(2,2)(-2,2)--(2,-2);}
}
\begin{tikzpicture}[
    x=.5cm, y=.5cm,
    >=latex,
    font= \footnotesize
    ]
    
\begin{scope}[xscale=0.5, yscale=1.7]
    
    
      \draw[->,thin] (-10,0) -- (11,0) node[right, below] { $x$}; 
      \draw[->,thin] (0,-1) -- (0,5) node[above, left] { };
    
    \foreach \x in {-8,-6,...,-2,0,2,...,10}
    \draw (\x,-.1) -- (\x,.1) node[below=10pt] {$\x$};
    
    \foreach \y in {	}
    \draw (-.1,\y) -- (.1,\y) node[left=10pt] {$\y$};

     \draw[-, line width = 2pt, draw = Mahogany] (-4.8,1)--(3.2,1) node[pos=0.5,sloped,above] {$[-4.8,3.2]$};
     \draw[-, line width = 2pt, draw = Mahogany] (9.5,2)--(10.5,2) node[pos=0.5,sloped,above] {$[9.5,10.5]$};
     \draw[-, line width = 2pt, draw = Mahogany] (-7.2,3)--(-5.2,3) node[pos=0.5,sloped,above] {$[-7.2,-5.2]$};
     \draw[-, line width = 2pt, draw = Mahogany] (3.2,4)--(7.2,4) node[pos=0.5,sloped,above] {$[3.2,7.2]$};

    \draw[line width=1pt, decorate,
    decoration=crosses]
    (-0.8, 1) -- (-0.801, 1) node[pos=0.5,sloped,below] {$-0.8$};
    
    \draw[line width=1pt, decorate,
    decoration=crosses]
    (0.4, 2) -- (0.401, 2) node[pos=0.5,sloped,below] {$0.4$};
    
    \draw[line width=1pt, decorate,
    decoration=crosses]
    (-0.8, 3) -- (-0.801, 3) node[pos=0.5,sloped,below] {$-0.8$};
    
    \draw[line width=1pt, decorate,
    decoration=crosses]
    (0.8, 4) -- (0.801, 4) node[pos=0.5,sloped,below] {$0.8$};
    
    
      
\end{scope}
\end{tikzpicture}

&&
        \begin{tikzpicture}[
    x=.5cm, y=.5cm,
    >=latex,
    font= \footnotesize
    ]
    
\begin{scope}[xscale=0.6, yscale=2.0]
    
    
      \draw[->,thin] (-7,0) -- (8,0) node[right, below] { $x$}; 
      \draw[->,thin] (0,-0.8) -- (0,3) node[above, left] { };
    
    \foreach \x in {-6,-4,...,-2,0,2,...,6}
    \draw (\x,-.1) -- (\x,.1) node[below=10pt] {$\x$};
    
    \foreach \y in {	}
    \draw (-.1,\y) -- (.1,\y) node[left=10pt] {$\y$};

     \draw[-, line width = 2pt, draw = Mahogany] (-5,1)--(3,1) node[pos=0.5,sloped,above] {$[-5,3]$};
     \draw[-, line width = 2pt, draw = Mahogany] (3,2)--(7,2) node[pos=0.5,sloped,above] {$[3,7]$};

     \draw[line width=1pt, decorate,
     decoration=crosses]
     (-1, 1) -- (-1.01, 1) node[pos=0.5,sloped,below] {$-1.0$};
     
     \draw[line width=1pt, decorate,
     decoration=crosses]
     (1, 2) -- (1.01, 2) node[pos=0.5,sloped,below] {$1.0$};

    
      
\end{scope}
\end{tikzpicture}\\
      {\small (a) An optimal placement of} & &
      {\small (b) An optimal placement of}\\
      {\small a $4$-system $S' = \{ 1,2,4,8\}$.} & &
      {\small a $2$-system $S'' = \{ 4,8 \}$.}
    \end{tabular}
  \end{center}
  \vspace*{-12pt}
  \caption{Additional intervals may improve the variation of a small placement by involving gaps between the placed intervals.}
  \label{fig:captureTheSubpath}
\end{figure}
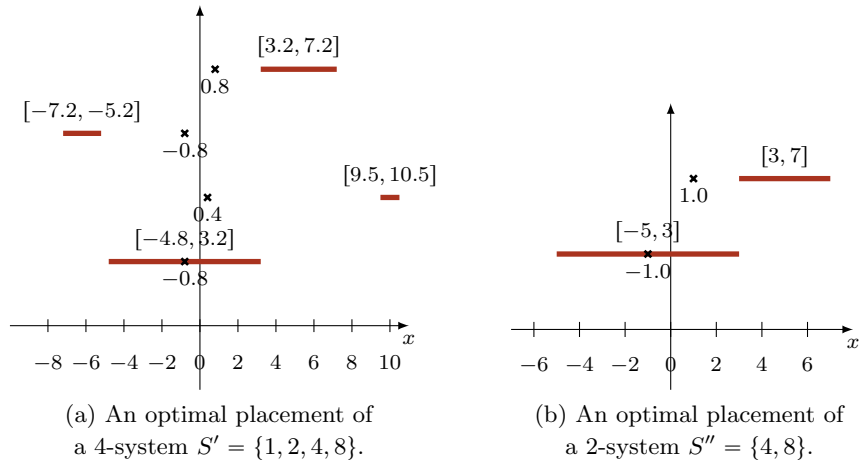
	
	Let $m_i$, $m_i'$, and $m_i''$ be the midpoints of the placement illustrated in Figure~\ref{fig:captureTheSubpath}. It is important to observe that although $S'' \supset S'$, we have $m_1 \neq m'_1$ and $m_2 \neq m'_2$. In particular, the usage of the additional (smaller) intervals allow a different placement $p''$ that has a smaller deviation than~$p'$. The high-level reason for this improvement is that the intervals corresponding to $\ell_3$ and $\ell_4$ are placed before $\ell_2$ and thus reduce the influence of $\ell_2$ to the deviation of the placement illustrated in Figure~\ref{fig:captureTheSubpath}(b). In particular, the deviation of the placements illustrated in Figure~\ref{fig:captureTheSubpath} determined by sum of the lengths of the two largest intervals. Furthermore, the deviation is decreased by the sum of the lengths of all intervals that are considered. Thus we chose $\tau$ as follows. Let $\ell'$ and $\ell''$ be the largest and second largest lengths of a heterogenous system $S = \{ \ell_1,\dots,\ell_n \}$. Furthermore, let $\ell = \min \{ \ell_1,\dots,\ell_n \}$. By applying that the lengths of $\ell_1,\dots,\ell_n$ increase constantly by a factor of $x \geq 2$, we make the following observation.
	
	\begin{observation}\label{obs:tau}
		$\tau = \frac{\ell x^{2n-4}(x^2-1)}{4 \sum_{i=1}^n \ell_i}$.
	\end{observation}
	
	W.l.o.g., we assume $m_1 \leq 0$. The argument for the case $m_1 \geq 0$ is symmetric. Furthermore, w.l.o.g., we assume $\ell_1 = |I_1|,\dots,\ell_n = |I_n|$ such that $\langle \ell_1,\dots,\ell_n \rangle$.
	
	In the following we show that $\Delta(p) \geq \tau$ holds if $x \geq 2$. To this end, we first prove that the largest interval has to be placed first in an optimal placement, see Lemma~\ref{lem:helplargestfirstlowerbound}. Based on that, we establish that $\tau$ is a lower bound for the deviation of heterogenous systems, see Lemma~\ref{lem:tightbound}.

	As $S$ is heterogenous, there is a unique largest length $\ell_i \in \{ \ell_1,\dots,\ell_n\}$ such that $\ell_i = \ell x^{n-1}$. 

	 The following lemma verifies that an optimal placement $p$ places the largest interval $\ell_i$ first, if $\tau$ is a tight lower bound.

\begin{lemma}\label{lem:helplargestfirstlowerbound}
	If the longest block $\ell_i$ is not placed first, we have $\Delta(p) > \tau$.
\end{lemma}
\begin{proof}
	Note that $\ell_1,\dots,\ell_n$ is the order that corresponds to the placement $p$. Suppose $\ell_i$ is not placed first, i.e., that $\ell_i \neq \ell_1$ holds. Based on that, we show that the center $\cog{s_i}$ of the state $s_i$ is larger than $\tau$ if $x > \frac{1 + \sqrt{5}}{2}$ holds. As we are considering a heterogenous system, we have $x \geq 2 > \frac{1 + \sqrt{5}}{2}$, concluding the proof. 
	
	By applying Lemma~\ref{lem:currentcenterVSpreviouscenter} we can reformulate the statement to be shown, i.e. $|\cog{s_i}| > \tau$, as follows.
			\begin{equation}
				\left| \frac{\cog{s_{i-1}} \sum_{j=1}^{i-1} \ell_j + m_i \ell_i}{\sum_{j=1}^{i} \ell_j} \right| > \tau. \label{eq:tech}
			\end{equation}
	
	In order to show that Inequality~\ref{eq:tech} holds, we distinguish between two cases: (1) $m_i < m_1$ and (2) $m_i > m_1$. In both cases, an application of Lemma~\ref{lem:helplargestfirstlowerbound} concludes the proof as follows.
	
	\begin{itemize}
		\item  $m_i < m_1$: Lemma~\ref{lem:helptechlowerbound} implies $m_i \geq - \tau + \frac{\ell_1+\ell_i}{2} > - \tau + \frac{\ell_i}{2}$. By definition of $\tau$ we obtain $\tau < \frac{\ell_i}{2}$. Thus, we have $m_i > - \tau + \frac{\ell_i}{2} > 0$, which implies that the summand $m_i \ell_i$ in Equation~\ref{eq:tech} is positive. W.l.o.g., we assume $m(i) = - \tau + \frac{\ell_i}{2}$, $i = n$, and $\cog{s_{i-1}} = -\tau$, because this does not increase the left side of Inequality~\ref{eq:tech}. Hence, we obtain the following.
			\begin{alignat*}{6}
				& \frac{- \tau \sum_{j=1}^{n-1} \ell_j \left( -\tau + \frac{\ell_n}{2} \right)}{\sum_{j=1}^{n} \ell_j} &\;& >&\; & \tau\\
				\Leftrightarrow & - \tau \sum_{j=1}^{n} \ell_j + \frac{\ell_n^2}{2} &\;& > &\;& \tau \sum_{j=1}^{n} \ell_j\\
				\Leftrightarrow & \frac{\ell_n^2}{2} &\;& > &\;& 2 \tau \sum_{j=1}^n \ell_j
			\end{alignat*}
		By applying Observation~\ref{obs:tau} and $\ell_n = \ell x^{n-1}$, we reformulate this as follows:
			\begin{alignat*}{6}
				\Leftrightarrow & \frac{\left( \ell x^{n-1} \right)^2}{2} &\;& > &\;& \frac{\ell x^{n-2} \left( \ell x^{n-1} + \ell x^{n-2} \right)}{2}\\
				\Leftrightarrow & x^{2n-4} \left( x^2 - x -1 \right) &\;& > &\;& 0\\
				\Leftrightarrow & x &\;& > &\;& \frac{1 + \sqrt{5}}{2}.
			\end{alignat*}
		\item $m_i > m_1$: This case is analogous to the previous case. 
	\end{itemize}
\end{proof}

	In the proof of Lemma~\ref{lem:helplargestfirstlowerbound}, we apply the following auxiliary lemmas: Lemma~\ref{lem:helptechlowerbound}, Lemma~\ref{lem:currentcenterVSpreviouscenter}, and Lemma~\ref{lem:positionunique}.
	
\begin{lemma}\label{lem:helptechlowerbound}
If $m_i > m_1$, we have $m_i \geq - \Delta(p) + \frac{\ell_1+\ell_i}{2}$. Otherwise, we have $m_i \leq -\frac{\ell_1+\ell_i}{2}$.
\end{lemma}
\begin{proof}
	Suppose $m_i > m_1$. As the intervals are pairwise disjoint, we have $m_i - m_1 \geq \frac{\ell_1 + \ell_i}{2}$, which is equivalent to $m_i \geq m_1 + \frac{\ell_1 + \ell_i}{2}$. By assumption we know $m_1 \leq 0$. Thus, we have $-\Delta(p) \leq m_1 \leq 0$. This implies $m_i \geq - \Delta(p) + \frac{\ell_1+\ell_i}{2}$.
	
	Now assume $m_i < m_1$. As the intervals are pairwise disjoint, we have $m_1 - m_i \geq \frac{\ell_1 + \ell_i}{2}$, which is equivalent to $m_i - m_1 \leq - \frac{\ell_1 + \ell_i}{2}$. This implies $m_i \leq - \frac{\ell_1 + \ell_i}{2}$, because $-m_1 \geq 0$.

\begin{figure}[h!]
  \begin{center}
       \begin{tikzpicture}[
    x=.5cm, y=.5cm,
    >=latex,
    font= \footnotesize
    ]
    
\begin{scope}[xscale=1.0, yscale=2.5]
    
    
      \draw[->,thick] (-7,0) -- (10.5,0) node[right, below] {}; 
      \draw[->,thick] (0,-1) -- (0,4) node[above, left] {};
    
    

     \draw[-, line width = 2pt, draw = Mahogany] (-7,1)--(3,1) node[pos=0.5,sloped,above] {$$};
     \draw[-, line width = 2pt, draw = Mahogany] (3,3)--(7,3) node[pos=0.5,sloped,above] {$$};
     \draw[-, line width = 2pt, draw = Mahogany] (8,2)--(10,2) node[pos=0.5,sloped,above] {$$};
     \draw[-, line width = 1pt, draw = MidnightBlue, style = dashed] (-2.5,0)--(-2.5,3.7) node[pos=0.5,sloped,right, rotate = 270] {$\tau-  $};
     \draw[-, line width = 1pt, draw = MidnightBlue, style = dashed] (2.5,0)--(2.5,3.7) node[pos=0.5,sloped,left, rotate = 270] {$\tau  $};

     \draw[line width=1pt, decorate,
     decoration=shape backgrounds]
     (-2, 1) -- (-2.01, 1) node[pos=0.5,sloped,below] {$\qquad m_1 \ge -\tau$};

          \draw[line width=1pt, decorate,
          decoration=shape backgrounds]
          (5, 3) -- (5.01, 3) node[pos=0.1,sloped,below] {$\qquad \qquad \qquad \quad m_3 \ge -\tau + \frac{l_1}{2} + \frac{l_3}{2}$};

    
      
\end{scope}
\end{tikzpicture}
  \end{center}
  \vspace*{-12pt}
  \caption{Estimation of the position of an interval $I_3$ that is not placed first and to the right of the longest interval.}
  \label{fig:helptechlowerbound}
\end{figure}
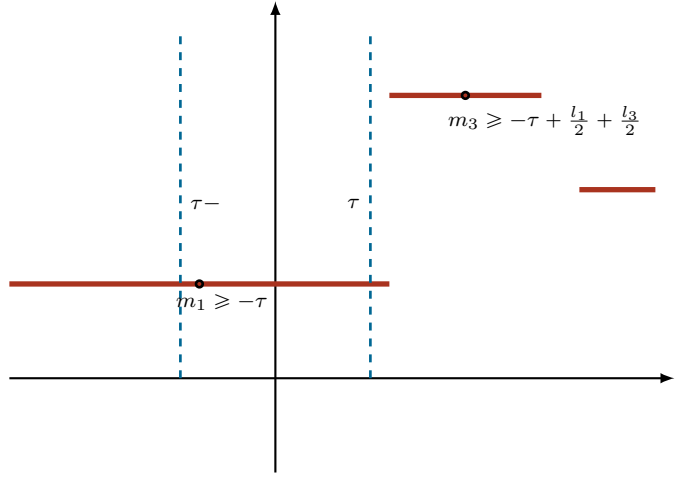
\end{proof}

	For the two following lemmas, we consider an arbitrarily chosen but fixed $k \in \{2,\dots,n \}$ and abbreviate $r:=s_{k-1}$ and $s := s_{k}$.
	
	The following lemma describes how $\cog{s}$ can be formulated in terms of $\cog{r}$. 

\begin{lemma}\label{lem:currentcenterVSpreviouscenter}
		$\cog{s} = \frac{\cog{r} \left( \sum_{i =1}^{k-1} \ell_i \right) + m_k \ell_{k}}{\sum_{i =1}^k \ell_i}$.
\end{lemma}
\begin{proof} By applying the definition of the center of gravity, we obtain the following.
	\begin{eqnarray*}
		\cog{s} & =& \frac{\sum_{i =1}^k m_i \ell_i}{ \sum_{i =1}^{k-1} \ell_i}\\
		& =& \frac{\sum_{i =1}^{k-1} m_i \ell_i + m_k \ell_{k}}{\sum_{i =1}^k \ell_i} \\
		&=& \frac{\left( \frac{\sum_{i =1}^{k-1} \ell_i}{\sum_{i =1}^{k-1} \ell_i} \right) \sum_{i =1}^{k-1} m_i \ell_i + m_k \ell_{k}}{\sum_{i =1}^{k-1} \ell_i}\\
		& = &  \frac{ \cog{r} \left( \sum_{i =1}^{k-1} \ell_i \right) + m_k \ell_{k}}{\sum_{i =1}^{k-1} \ell_i}
	\end{eqnarray*}
\end{proof}

	The following lemma describes how the \gravityc s $\cog{r}$ and $\cog{s}$ uniquely determine the midpoint $m_k$ of the $k$-th interval.
	
\begin{lemma}\label{lem:positionunique}
		$m_k = \frac{\cog{s} \sum_{i=1}^{k} \ell_i - \cog{r} \sum_{i =1}^{k-1} \ell_i}{\ell_{k}}$.
\end{lemma}
\begin{proof}
	By combining the definition of $\cog{s}$ and Lemma~\ref{lem:currentcenterVSpreviouscenter}, we obtain the following.
	\begin{alignat*}{4}
		& \cog{s} && = &\;& \frac{\sum_{i =1}^k m_i \ell_i}{\sum_{i =1}^k\ell_i} \\
		&&&= &\;& \frac{\cog{r} \sum_{i =1}^{k-1} \ell_i + m_k \ell_{k}}{\sum_{i =1}^k \ell_i}\\
		\Leftrightarrow & \cog{s} \sum_{i =1}^k \ell_i && = &\;& \cog{r} \sum_{i =1}^{k-1} \ell_i + m_k \ell_{k}\\
		\Leftrightarrow & m_k & &= &\;&\frac{\cog{s} \sum_{i =1}^k \ell_i - \cog{r} \sum_{i =1}^{k-1} \ell_i}{\ell_{k}}
	\end{alignat*}
\end{proof}

\begin{lemma}\label{lem:techhelplowerboundhetergenous}
	If $\ell_1 = \ell x^{n-1}$, $\ell_k = \ell x^{n-2}$, $m_1 < m_k$, and $\Delta(p) \leq \frac{\ell_k}{4}$. Then we have $m_k \geq \cog{s_{k-1}} + \frac{\ell_1+\ell_k}{2}$.
\end{lemma}
\begin{proof}
	Suppose $m_k < \cog{s_{k-1}} + \frac{\ell_1+\ell_k}{2}$. Observation~\ref{obs:interval} implies $|m_1| = |\cog{s_1}| \leq \Delta(p) \leq \frac{\ell_k}{4}$. This implies $-\frac{\ell_k}{4} \leq m_1 \leq 0$ because $m_1 \leq 0$. As $m_k > m_1$, we have $m_k \geq m_1 + \frac{\ell_1 + \ell_k}{2}$. This implies $m_k \geq -\frac{\ell_k}{4} + \frac{\ell_1 + \ell_k}{2} \geq -\frac{\ell_k}{4} + \frac{3\ell_k}{2} > \frac{3 \ell_k}{4}$ because $\ell_1 \geq 2 \ell_k$. Furthermore, Lemma~\ref{lem:currentcenterVSpreviouscenter} implies 
	\begin{equation*}
		n\cog{s_{k}} = \frac{\cog{s_{k-1}} \sum_{j=1}^{k-1} \ell_j + m_k \ell_k}{\sum_{j=1}^{k} \ell_j}.
	\end{equation*}
	
	By combining the above, the lemma follows by contradiction to $\frac{\ell}{4} \geq \Delta(p) \geq |\cog{s_k}|$.
	\begin{eqnarray*}
		\cog{s_k} & \stackrel{\text{Lemma}~\ref{lem:currentcenterVSpreviouscenter}}{=} & \frac{\cog{s_{k-1}} \sum_{j=1}^{k-1} \ell_j + m_k \ell_k}{\sum_{j=1}^{i} \ell_j}\\
		& \stackrel{\text{Observation}~\ref{obs:interval}}{>} & \frac{-\frac{\ell_k}{4} \sum_{j=1}^{k-1} \ell_j + m_k \ell_k}{\sum_{j=1}^{k} \ell_k}\\
		& =& \frac{-\frac{\ell_k}{4} \sum_{j=1}^{k} \ell_j}{\sum_{j=1}^k \ell_j} + \frac{\frac{\ell_k}{4} \ell_k}{\sum_{j=1}^k \ell_j} + \frac{m_k \ell_k}{\sum_{j=1}^k \ell_j}\\
		& \stackrel{\sum_{j=1}^{k} \ell_k \leq 2 \ell_k}{\geq} & -\frac{\ell_k}{4} + \frac{\frac{\ell_k}{4} \ell_k}{\ell_k} + \frac{m_k}{2}\\
		& \stackrel{m_k \geq \frac{3 \ell_k}{4}}{\geq} & \frac{\ell_k}{4}.
	\end{eqnarray*}
\end{proof}

	Based on Lemma~\ref{lem:helplargestfirstlowerbound}, we prove that $\tau$ is a lower bound for the deviation of any placement~$p$.

Now we are ready to prove Lemma~\ref{lem:tightbound}.

\begin{proof}
	W.l.o.g., we assume that the largest block is placed first, i.e. $\ell_i = \ell x^{n-1}$. Otherwise, Lemma~\ref{lem:helplargestfirstlowerbound} implies $\Delta(p) \geq \tau$. 
	
	Let $k \in \{ 1,\dots,n \}$ be the index such that $\ell_k = \ell x^{n-2}$ is the second largest interval. Recall that $m(1) \leq 0$ and distinguish two cases: $m_k > m_i$ and $m_k < m_i$.
	
	\begin{itemize}
		\item $m_k > m_1$: The definition of $\cog{s_{k-1}}$ implies $\sum_{j=1}^{k-1} m(j) \ell_{j} = \cog{s_{k-1}} \sum_{j=1}^{k-1} \ell_j$ $(\star)$. By combining $(\star)$, the definition of the center of gravity, and Lemma~\ref{lem:techhelplowerboundhetergenous} we can show the lemma as follows:
			\begin{alignat*}{4}
				& \cog{m_{k}} &\;=\; & \frac{\sum_{j=1}^{k} m_j \ell_j}{\sum_{j=1}^{k} \ell_j}\\
				\Leftrightarrow \; & \sum_{j=1}^{k-1} m_j \ell_j + m_k \ell_k & \;=\; & \cog{s_{k}} \sum_{j=1}^{k} \ell_j\\
				 \stackrel{(\star)}{\Rightarrow} \; & \cog{m_{k-1}} \sum_{j=1}^{k-1} \ell_j + m_k\ell_k & \;=\;& \cog{s_k} \sum_{j=1}^{k} \ell_j
			\end{alignat*}
		By applying Lemma~\ref{lem:techhelplowerboundhetergenous} we get the following.
			\begin{alignat*}{4}
				 & \cog{s_{k-1}} \sum_{j=1}^{k-1} \ell_k\\& + \left( \cog{s_{k-1}} + \frac{\ell_1 + \ell_k}{} \right)\ell_k & \;\leq\; & \cog{s_{k}} \sum_{j=1}^k \ell_j&\\
				 \Leftrightarrow \;\;\;\;\; & \cog{s_{k-1}} \sum_{j=1}^{k} \ell_j + \frac{\ell_1 + \ell_k}{2}\ell_k & \;\leq\; & \cog{s_{k}} \sum_{j=1}^{k} \ell_k\\
				 \Rightarrow \;\;\;\;\; & \frac{\ell_1+\ell_k}{2}\ell_k &\; \leq\; & 2 \max \{ \cog{s_{k-1}}, \cog{s_{k}} \} \sum_{j=1}^k \ell_j\\
				 \Leftrightarrow\;\;\;\;\; & \frac{\ell_1 + \ell_k}{4 \sum_{j=1}^k \ell_j} \ell_k &\;\leq\; & \max \{  \cog{s_{k-1}}, \cog{s_{k}}\}\\
				 \stackrel{\begin{array}{c}
 					\ell' = \ell_1,\\
					\ell'' = \ell_k
				\end{array}
				}{\Rightarrow} & \frac{\ell'+\ell''}{4 \sum_{j=1}^k \ell_j} \ell_k & \;\leq\; & \Delta(p).
			\end{alignat*}
		\item $m_k < m_1$: Combining $m_1 \leq 0$, $m_k$, and $|m_1m_k| \geq \frac{\ell_1}{\ell_k}$ leads to $m_k < -\frac{\ell_1 + \ell_k}{2} < -\frac{\ell_k}{4}$. This implies $\Delta(p) > \frac{\ell_k}{4} = \frac{\ell_1 + \ell_k}{4 (\ell_1 + \ell_k)} \ell_k \geq \frac{\ell_1 + \ell_k}{\sum_{j=1}^k \ell_j} \ell_k = \tau$ as claimed. 
	\end{itemize}
\end{proof}

\subsection{Proof of Lemma~\ref{lem:strategy1optimal}}\label{sec:prooflemstrategy1optimal}

\statement{Lemma}{lem:strategy1optimal}
               {\em
               		We have $\Delta(p) = \tau$ for the placement $p$ computed by the algorithm described in Section~\ref{sec:heterogenous}.
}
\begin{proof}
	 We have $\cog{s_0} = 0$. Furthermore, by Observation~\ref{obs:interval}, we obtain~$\cog{s_1} = - \tau$. In the following we distinguish between two cases: (1) $n$ is even and (2) $n$ is odd.
	 
	 \begin{itemize}
	 	\item (1) $n$ is even: Lemma~\ref{lem:positionunique} implies $m_2 = \frac{\cog{s_2} (\ell_1 + \ell_2) - \cog{s_1} \ell_1}{\ell_2}$. Combining this with $\cog{s_1} = -\tau$ implies $\cog{s_2} = \tau$ as follows:
			\begin{alignat*}{4}
				&m_2 &\;=\;& \frac{\cog{s_2} (\ell_1 + \ell_2) - \cog{s_1} \ell_1}{\ell_2}\\
				\stackrel{\phantom{\cog{s_1} = \tau}}{\Rightarrow}&\left( -1 \right)^2 \left( \frac{2 \tau \ell_1}{\ell_2} + \tau \right)&\;=\;& \frac{\cog{s_2} (\ell_1 + \ell_2) - \cog{s_1} \ell_1}{\ell_2}\\
				\stackrel{\cog{s_1} = \tau}{\Leftrightarrow} & \frac{\tau\ell_1 + \tau(\ell_1 + \ell_2)}{\ell_2}&\;=\;&\frac{\cog{s_2} (\ell_1 + \ell_2) +\tau \ell_1}{\ell_2}\\
				\stackrel{\phantom{\cog{s_1} = \tau}}{\Leftrightarrow} & \tau &\;=\;&\cog{s_2}.
			\end{alignat*}
			Let $i \in \{ 3,\dots,n \}$. In the following we show $\cog{s_i} = \tau$ if $i$ is even and $\cog{s_i} = -\tau$ if $i$ is odd. Suppose $\cog{s_i} = \tau$ holds for all even $j \in \{ 4,\dots,i-1 \}$ $(\dagger)$ and $\cog{s_i} = -\tau$ holds for all odd $j \in \{ 3,\dots,i-1 \}$ $(\ddagger)$.
			
			 We first consider the case that $i$ is even. Lemma~\ref{lem:positionunique} implies that $m_i$ is equal to 
			 
			 $$\frac{\cog{s_i} \sum_{j=1}^i \ell_i - \cog{s_{i-1}} \sum_{j=1}^{i-1} \ell_j}{\ell_i}$$
			 
			 The algorithm guarantees $m_i = \frac{2 \tau \sum_{k=1}^{i-1} \ell_k}{\ell_i} + \tau $. Furthermore, $(\ddagger)$ implies $\cog{s_{i-1}} = - \tau$. Combining the above three equations yields $\cog{s_i} = \tau$ as follows:
			 \begin{alignat*}{4}
			 	&\frac{2 \tau \sum_{k=1}^{i-1} \ell_k}{\ell_i} + \tau &\;=\;&\frac{\cog{s_i} \sum_{j=1}^i \ell_i + \tau \sum_{j=1}^{i-1} \ell_j}{\ell_i}\\
				\Leftrightarrow\; & \frac{\tau \sum_{k=1}^{i-1} \ell_k + \tau \sum_{k = 1}^{i} \ell_k}{\ell_i} &\; =\;& \frac{\cog{s_i} \sum_{j=1}^i \ell_i + \tau \sum_{j=1}^{i-1} \ell_j}{\ell_i}\\
				\Leftrightarrow\; & \tau &\;=\;&\cog{s_i}.
			 \end{alignat*}
			By applying a similar approach for odd $i$, we also obtain $\cog{s_i} = -\tau$.
			
			By induction it follows $|\cog{s_i}| = \tau$ for all $i \in \{ 3,\dots,n \}$ if $n$ is even and thus we have $\Delta(p) = \tau$ if $n$ is even.
		\item (2) $n$ is odd: By the definition of the algorithm we have $\cog{s_1} < \cog{s_2}$ and $\cog{s_2} < \cog{s_3}$. In the following we show $\cog{s_3} = \tau$. This implies $|\cog{s_2}| \leq \tau$ because $-\tau = \cog{s_1} < \cog{s_2}$. Furthermore, a similar approach as in the first case implies $\cog{s_i} = \tau$ if $i$ is odd and $\cog{s_{i}} = -\tau$ if $i$ is even. Thus we obtain $\Delta(p) \leq \tau$ if $n$ is odd.
		
			Finally, we show $\cog{s_3} = \tau$. By the definition of the algorithm we know that the intervals that correspond to $\ell_2$ and $\ell_3$ are placed adjacently on the right side of the interval that corresponds to~$\ell_1$. Thus, for estimating $\cog{s_3}$ we are allowed to consider the two intervals $\ell_2$ and $\ell_3$ as one interval. Let $q$ be the midpoint of this interval. Hence, Lemma~\ref{lem:positionunique} implies 
			\begin{equation*}
				q = \frac{\cog{s_1} \ell_1 + \cog{s_3}(\ell_1+\ell_2+\ell_3)}{\ell_2+\ell_3}.
			\end{equation*}
			
			Furthermore, the algorithm guarantees
			
			\begin{equation*}
				q = \frac{2 \tau (\ell_1 + \ell_2 + \ell_3)}{\ell_3} + \tau.
			\end{equation*}
			
			Combining the two last equations with $\cog{s_1} = - \tau$ leads to $\cog{q} = \tau$ as follows.
				\begin{alignat*}{4}
					&\frac{\cog{s_3}(\ell_1+\ell_2+\ell_3) - \cog{s_1} \ell_1}{\ell_2+\ell_3}&\;=\;&\frac{2 \tau \ell_1 }{\ell_2+\ell_3} + \tau\\
					\Leftrightarrow\; & \frac{\cog{s_3}(\ell_1+\ell_2+\ell_3) + \tau \ell_1}{\ell_2+\ell_3}&\;=\;&\frac{\tau \ell_1 + \tau (\ell_1+ \ell_2 + \ell_3)}{\ell_2+\ell_3}\\
					\Leftrightarrow\; & \cog{s_3}&\;=\;&\tau.
				\end{alignat*}
			As $\cog{q} = \cog{s_3}$, we obtain $\cog{s_3} = \tau$. This concludes the proof.
	 \end{itemize}
\end{proof}

\subsection{Proof of Lemma~\ref{lem:validplacementheterogenous}}\label{sec:prooflemvalidplacementheterogenous}

\statement{Lemma}{lem:validplacementheterogenous}
               {\em
               		The intervals as computed by the algorithm from above are pairwise disjoint.
}

\medskip
In the following, we give a proof for Lemma~\ref{lem:validplacementheterogenous}. In particular, let $I_1,\dots,I_n \subset \mathbb{R}$ be the intervals of lengths $\ell_1,\dots,\ell_n$ that are computed by the algorithm. For two intervals, $I_i$ and $I_j$, we abbreviate $I_i \leq I_j$ if $m_i < m_j$ and $|m_i - m_j| \geq \frac{\ell_i + \ell_j}{2}$. In the following, we show that two intervals do not overlap, i.e., that $I_i \leq I_j$ or $I_j \leq I_i$ holds for all $i \neq j \in \{ 1,\dots,n \}$. We prove this separately for odd $n \geq 7$ and even $n \geq 6$ and explicitly for $n=4$ and $n=5$.

\begin{lemma}\label{lem:placevalidngeq6}
	Let $S = \langle \ell_1,\dots,\ell_n \rangle$ be a heterogeneous system for an even $n \geq 6$ and $p = \langle I_1,\dots,I_n \rangle$ the placement that is computed by the our algorithm. Then the intervals from $p$ are pairwise disjoint if
	\begin{itemize}
		\item (S1.1): $x^{n+7} + x^{n+3} + x^5 + x^4 + x^2 + 1 \geq 2x^{n+5} + x^{n+2} + x^n + x^7 + x^6$,
		\item (S1.2): $x^{n+5} + x^{n+2} + x^{n+1} + x^4 + x^2 \geq 2x^{n+4} + x^n + x^5 + x^1$, and
		\item (S1.3): $x^{n+5} + x^{n+1} + x^3 + 2x^2 + 1 \geq 2 x^{n+3} + x^{n+2} + x^n + x^5 + x^4$.
	\end{itemize}
\end{lemma}
\begin{proof}
	In the following we show that $I_3 < I_5 < \dots < I_{n-1} < I_1 < I_n < I_{n-2} < I_{n-4} < \dots < I_2$ holds. In order to do this we prove three implications.
	
	\begin{itemize}
		\item (S1.1) implies $I _3 \leq I_5 \leq \dots \leq I_{n-1}$,
		\item (S1.2) implies $I_{n-1} \leq I_1$, and
		\item (S1.3) implies $I_n \leq I_{n-2} \leq I_{n-4} \leq \dots \leq I_2$.
	\end{itemize}
	Furthermore, the inequality $I_1 \leq I_n$ is correct by the definition of $\tau$. This concludes the proof.
	\begin{itemize}
		\item (S1.1) implies $I _3 \leq I_5 \leq \dots \leq I_{n-1}$: Let $i \in \{ 3,5,\dots,n-3 \}$ be chosen arbitrarily. We have $I_i \leq I_{i+2}$ if $m_i + \frac{\ell_i}{2} \leq m_{i+2} - \frac{\ell_{i+2}}{2}$ holds. Furthermore, above we already showed $\cog{s_{i-1}} = \cog{s_{i+1}} = \tau$ and $\cog{s_i} = \cog{s_{i+2}} = -\tau$. The algorithm guarantees
			\begin{equation*}
				m_i = (-1)^i \left( \frac{2 \tau \sum_{j=1}^{i-1} \ell_j}{\ell_i} + \tau \right)
			\end{equation*}
			and 
			\begin{equation*}
				m_{i+2} = (-1)^{i+2} \left( \frac{2 \tau \sum_{j=1}^{i+1} \ell_j}{\ell_{i+2}} + \tau \right).
			\end{equation*}
			Thus, we formulate (S1.1) as a sufficient condition for $m_i + \frac{\ell_i}{2} \leq m_{i+2} - \frac{\ell_{i+2}}{2}$ as follows by applying $\ell_{i+2} = x^2 \ell_i$ $(\star)$ and the geometric sum $(\dagger)$.
			\begin{alignat*}{4}
				& m_i + \frac{\ell_i}{2} & \;\leq\; & m_{i+2} - \frac{\ell_{i+2}}{2}\\
				\Leftrightarrow\;&\frac{-2 \tau \sum_{j=1}^{i-1} \ell_j}{\ell_i} - \tau + \frac{\ell_i}{2} & \;\leq\; & \frac{-2 \tau \sum_{j=1}^{i+1} \ell_j}{\ell_{i+2}} - \tau - \frac{\ell_{i+2}}{2}\\
				\Leftrightarrow\; & 2 \tau \left( \frac{\sum_{j=1}^{i+1} \ell_j}{\ell_{i+2}} - \frac{\sum_{j=1}^{i-1} \ell_j}{\ell_i} \right) & \;\leq\; & - \frac{1}{2} \left( \ell_{i+2} + \ell_i \right)\\
				\stackrel{(\star)}{\Leftrightarrow}\; & 2 \tau \left( \frac{1}{x^2} \sum_{j=1}^{i+1} \ell_j - \sum_{j=1}^{i-1} \ell_j \right) & \;\leq\; & - \frac{1}{2} \left( x^2 \ell_i + \ell_i \right) \ell_i\\
				\Leftrightarrow\; & 2 \tau \left( \frac{1}{x^2} \ell_1 - \ell_1 + \frac{1}{x^2} \sum_{j=2}^{i+1} \ell_j - \sum_{j=2}^{i-1} \ell_j \right) & \;\leq\; & - \frac{x^2 + 1}{2} \ell_i^2\\
				\stackrel{(\dagger)}{\Leftrightarrow}\; & 2 \tau \left(\begin{matrix} \frac{\ell}{x^2} x^{n-1} - \ell x^{n-1}\\ \\+ \frac{\ell}{x^2} \left( \frac{1-x^i}{1-x} \right) - \ell \left( \frac{1 - x^{i-2}}{1-x} \right) \end{matrix} \right) & \;\leq\; & - \frac{x^2 + 1}{2} \ell_i^2\\
				\Leftrightarrow\; & 2 \ell \tau \left( \begin{matrix}x^{n-3} - x^{n-1} \\ \\+ \frac{x^{-2} x^{i-2}}{1-x} - \frac{1 - x^{i-2}}{1-x}\end{matrix} \right) & \;\leq\; & - \frac{x^2 + 1}{2} \ell_i^2\\
				\Leftrightarrow\; & 2 \ell \tau \left( x^{n-3} - x^{n-1} + \frac{x^{-2} -1}{1-x} \right) & \;\leq\; & - \frac{x^2+1}{2} \ell_i^2.\\
			\end{alignat*}
			As $\ell_i$ is minimized for $i = n-3$, we substitute $\ell_i$ by $l_{n-3} = \ell x^{n-5}$. Furthermore, we substitute $\tau = \frac{\ell x^{2n-4}(x^2 -1)}{4(x^n-1)}$. Hence
			\begin{alignat*}{4}
				& 2 \ell \tau \left( x^{n-3} - x^{n-1} + \frac{x^{-2} -1}{1-x} \right) & \;\leq\; & - \frac{x^2+1}{2} \ell_i^2.\\
				\Leftrightarrow\; & 2 \ell \left( \frac{\ell x^{2n-4}(x^2-1)}{4(x^n-1)} \right) \left( x^{n-3} - x^{n-1} + \frac{x^{-2} -1}{1-x} \right) & \;\leq\; & - \left( x^2 + 1 \right) x^{2n-10}\\
				\Leftrightarrow\; & x^6 \left( x^2-1 \right) \left( x^{n-3} - x^{n-1} + \frac{x^{-2}-1}{1-x} \right) & \;\leq\; & - (x^2+1)(x^n-1)\\
				\Leftrightarrow\; & x^6(x^2-1)(x^{n-3} - x^{n-1}) + x^6(x^2-1) \left( \frac{x^{-1} -1}{1-x} \right) & \;\leq\; & -(x^2+1)(x^n-1)\\
				\Leftrightarrow\; & x^6(x^2-1)(x^{n-3} - x^{n-1}) - x^6(x+1) \left( x^{-2} -  \right) &\; \leq \;& -(x^2+1)(x^n-1)\\
				\Leftrightarrow\; & x^{n+7} + x^{n+3} + x^5 + x^4 + x^2 + 1 &\;\geq\;& 2x^{n+5} + x^{n+2}\\
				&&& + x^n + x^7 + x^6.
			\end{alignat*}
		
		\item (S1.2) implies $I_{n-1} \leq I_1$: (S1.2) can formulated as a sufficient condition for $I_{n-1} \leq I_1$ as follows: $I_{n-1} \leq I_1$ is equivalent to $m_{n-1} + \frac{\ell_{n-1}}{2} \leq m_1 - \frac{\ell_2}{2}$ which can be reformulated as follows by applying $\tau=\frac{\ell x^{2n-4}(x^2-1)}{4(x^n - 1)}$ $(\star)$ and the geometric sum $(\dagger)$.
			\begin{alignat*}{4}
				& m_{n-1} + \frac{\ell_{n-1}}{2} &\leq\;& m_1 - \frac{\ell_2}{2}\\
				\Leftrightarrow\; & \frac{-2 \tau \sum_{j=1}^{n-2} \ell_j}{\ell_{n-1}} - \tau + \frac{\ell_{n-1}}{2} & \;\leq\; & - \tau - \frac{\ell_1}{2}\\
				\stackrel{(\dagger)}{\Leftrightarrow}\; & -2 \tau \left( \ell \frac{1 - x^{n-3}}{1-x} + \ell x^{n-1} \right) &\; \leq\; & - \frac{\ell_{n-1}}{2} (\ell_1 + \ell_{n-1})\\
				\Leftrightarrow\; & -2 \tau \left( \frac{1 - x^{n-3}}{1-x} + x^{n-1} \right) &\; \leq\; & - \frac{\ell x^{n-3}}{2} \left( \ell x^{n-1} + \ell x^{n-3} \right)\\
				\stackrel{(\star)}{\Leftrightarrow}\; & -2 \ell \left( \frac{\ell x^{2n-4}(x^2-1)}{4(x^n - 1)} \right) \left( \frac{1-x^{n-3}}{1-x} + x^{n-1} \right) & \;\leq\; & \frac{-\ell^2x^{2n-4} - \ell^2x^{2n-6}}{2}\\
				\Leftrightarrow\; & - \left( \frac{x^{2n-4} (x^2-1)}{x^n-1} \right) \left( \frac{1 - x^{n-3}}{1-x} + x^{n-1} \right) & \;\leq\; & -x^{2n-4} - x^{2n-6}\\
				\Leftrightarrow\; & \left( -x^{2n - 2} + x^{2n-4} \right) \left( \frac{1 - x^{n-3}}{1- x} +x^{n-1} \right) & \;\leq\; & \left( \begin{matrix}-x^{2n-4} \\- x^{2n-6} \end{matrix} \right)(x^n - 1)\\
				\Leftrightarrow\; & x^{2n-7} (-x^5 + x^3) \left( \frac{1-x^{n-3}}{1-x} + x^{n-1} \right) &\; \leq\; & x^{2n-7} \left( -x^3-x \right)(x^n-1)\\
				\Leftrightarrow\; & \left( -x^5 + x^3 \right) \left( 1 - x^{n-3} + x^{n-1} - x^n \right) &\;\geq\;& \left( \begin{matrix} -x^{n+3} +x^3\\ - x^{n+1} + x \end{matrix}\right)(1-x)\\
				\Leftrightarrow\; & x^{n+5} + x^{n+2} + x^{n+1} + x^4 + x^2 & \;\geq\; & 2x^{n+4} + x^{n} + x^5 + x. 
			\end{alignat*}
		
		\item (S1.3) implies $I_n \leq I_{n-2} \leq I_{n-4} \leq \dots \leq I_2$: The proof for this statement is similar to the proof of the first statement. Let $i \in \{ 2,4,\dots,n-2 \}$ be chosen arbitrarily. $I_n \leq I_{n-2} \leq I_{n-4} \leq \dots \leq I_2$ is equivalent to $m_i - \frac{\ell_i}{2} \geq m_{i+2} + \frac{\ell_{i+2}}{2}$. We formulate (S1.3) as a sufficient condition for $m_i - \frac{\ell_i}{2} \geq m_{i+2} + \frac{\ell_{i+2}}{2}$ as follows: We have $\cog{s_{i-1}} = \cog{s_{i+1}} = -\tau$ and $\cog{s_{i}} = \cog{s_{i+2}} = \tau$. Thus 
			\begin{alignat*}{4}
				& m_i - \frac{\ell_i}{2} & \geq\; & m_{i+2} + \frac{\ell_i+2}{2}\\
				\Leftrightarrow\; & \frac{2 \tau \sum_{j=1}^{i-1} \ell_j}{ \ell_i} + \tau - \frac{\ell_i}{2} &\; \geq\; & \frac{2 \tau \sum_{j=1}^{i+1} \ell_j}{\ell_{i+1}} + \tau + \frac{\ell_{i+2}}{2}\\
				\Leftrightarrow\; & 2 \tau \left( \frac{\sum_{j=1}^{i-1} \ell_j}{\ell_i} - \frac{\sum_{j=1}^{i+1} \ell_j}{\ell_{i+2}} \right) & \;\geq\;  & \frac{1}{2} (\ell_{i+2} + \ell_i)\\
				\Leftrightarrow\; & 2 \tau \left( \frac{\sum_{j=1}^{i-1} \ell_j}{\ell_i} - \frac{\sum_{j=1}^{i+1} \ell_j}{x^2\ell_i} \right) & \;\geq\; & \frac{1}{2}(x^2 \ell_i + \ell_i)\\
				\Leftrightarrow\; & 2 \tau \left( \ell_1 - \frac{1}{x^2}\ell_1 + \sum_{j=1}^{i-1} \ell_j - \frac{1}{x^2} \sum_{j=2}^{i+1} \ell_j \right) & \;\geq\; & \frac{x^2 + 1}{2} \ell_i^2\\
				\Leftrightarrow\; & 2 \tau \left( \ell_1 - \frac{1}{x^2} \ell_1 + \sum_{j=2}^{i-1} \ell_j - \frac{1}{x^2} \sum_{j=2}^{i+1} \ell_j \right) & \;\geq\; & \frac{x^2+1}{2} \ell_i^2\\
				\Leftrightarrow\; & 2 \tau \left( \begin{matrix}\ell x^{n-1} - \frac{\ell}{x^2}x^{n-1} +\\ \ell \left( \frac{1 - x^{i-2}}{1-x} \right) - \frac{\ell}{x^2} \left( \frac{1 - x^i}{1-x} \right) \end{matrix}\right) &\;\geq\; &\frac{x^2+1}{2} \ell_i^2\\
				\Leftrightarrow\; & 2 \ell \tau \left( x^{n-1} - x^{n-3} + \frac{1-x^{i-2}}{1-x} - \frac{x^{-2} - x^{i-2}}{1-x} \right) & \;\geq\; & \frac{x^2 + 1}{2} \ell_i^2\\
				\Leftrightarrow\; & 2 \ell \tau \left( x^{n-1} - x^{n-3} + \frac{1-x^{-2}}{1-x} \right) & \;\geq\; & \frac{x^2 + 1}{2} \ell_i^2.
					\end{alignat*}
		As $\ell_i$ is minimized for $i=n-2$, we substitute $\ell_i$ by $\ell_{n-2 = \ell x^{n-4}}$. Furthermore, we substitute $\tau = \frac{\ell x^{2n-4}(x^2 -1)}{4(x^n-1)}$. Hence
			\begin{alignat*}{4}
				& 2 \ell \left( \frac{\ell x^{2n-4} (x^2 - 1)}{4 (x^n-1)} \right) \left( x^{n-1}  -x^{n-3} + \frac{1-x^{-2}}{1-x} \right) & \;\geq\; & \frac{x^2+1}{2} \left( \ell x^{n-4} \right)^2\\
				\Leftrightarrow\; & \frac{x^{2n-4} (x^2-1)}{x-1} \left( x^{n-1} - x^{n-3} + \frac{1-x^{-2}}{1-x} \right) & \;\geq\; & (x^2+1)(x^{2n-8})\\
				\Leftrightarrow\; & x^4 (x^2-1)\left( x^{n-1} - x^{n-3} + \frac{1-x^{-2}}{1-x} \right) & \;\geq\; & (x^2+1)(x^n-1)\\
				\Leftrightarrow\; & x^4 (x^2-1) (x^{n-1} - x^{n-3}) + x^4(x^2-1) \left( \frac{1-x^{-2}}{1-x} \right) &\;\geq\;& (x^2+1)(x^n-1)\\
				\Leftrightarrow\; & (x^6 - x^4)(x^{n-1} - x^{n-3})- x^4(x+1)(1-x^{-2}) & \;\geq\; & (x^2+1)(x^n-1)\\
				\Leftrightarrow\; & x^{n+5} + x^{n+1} + x^3 + 2x^2 + 1 &  \;\geq\; & 2x^{n+3} + x^{n+2} \\&&&+ x^n + x^5 + x^4.
			\end{alignat*}
	\end{itemize}
\end{proof}
	
\begin{lemma}\label{lem:placevalidneq4}
	Let $S = \langle \ell_1,\dots,\ell_n \rangle$ be a heterogeneous system for $n = 4$ and $p = \langle I_1,\dots,I_n \rangle$ the placement that is computed by the algorithm. Then the intervals from $p$ are pairwise disjoint if
	\begin{itemize}
		\item (S1.2): $x^{n+5} + x^{n+2} + x^{n+1} + x^4 + x^2 \geq 2x^{n+4} + x^n + x^5 + x^1$, and
		\item (S1.3): $x^{n+5} + x^{n+1} + x^3 + 2x^2 + 1 \geq 2 x^{n+3} + x^{n+2} + x^n + x^5 + x^4$.
	\end{itemize}
\end{lemma}
\begin{proof}
	Similar to the proof of Lemma~\ref{lem:placevalidngeq6}, we guarantee $I_3 \leq I_5 \leq \dots \leq I_{n-1} \leq I_1 \leq I_n \leq I_{n-2} \leq \dots \leq I_2$. As $n = 4$, we do not have to take care about $I_3 \leq \dots \leq I_{n-3} \leq I_{n-1}$. The proof for (S1.2) and (S1.3) is the same as in the proof of Lemma~\ref{lem:placevalidngeq6}.
\end{proof}

\begin{lemma}\label{lem:placevalidngeq7}
	Let $S = \langle \ell_1,\dots,\ell_n \rangle$ be a heterogeneous system for an odd $n \geq 7$ and $p = \langle I_1,\dots,I_n \rangle$ the placement that is computed by the algorithm. Then the intervals from $p$ are pairwise disjoint if
	\begin{itemize}
		\item (S1.1): $x^{n+7} + x^{n+3} + x^5 + x^4 + x^2 + 1 \geq 2x^{n+5} + x^{n+2} + x^n + x^7 + x^6$,
		\item (S1.2): $x^{n+5} + x^{n+2} + x^{n+1} + x^4 + x^2 \geq 2x^{n+4} + x^n + x^5 + x^1$,
		\item (S1.3): $x^{n+5} + x^{n+1} + x^3 + 2x^2 + 1 \geq 2 x^{n+3} + x^{n+2} + x^n + x^5 + x^4$, and
		\item (S1.4): $x^{2n}(x^{-2} - x^{-4})(x^{n+2} - x^n - x^{n-1} - x^3 - 2x^{2} - 2x - 1) \geq (x^3 + x + 1)(x^n -1)(x^3 + x^4)$.
	\end{itemize}
\end{lemma}
\begin{proof} (S1.1), (S1.2), and (S1.3) are the same conditions as in Lemma~\ref{lem:placevalidngeq6}. By considering the second and third interval as one interval, we are allowed to apply an approach that is similar to the argument from the proof of Lemma~\ref{lem:placevalidngeq6}. In particluar, except for the first three intervals we still have the property that lengths of the intervals increase by a factor of $x \geq 2$.

	We show that $I_3 \leq I_5 \leq \dots \leq I_{n-4} \leq I_{n-2} \leq I_1 \leq I_{n-1} \leq I_{n-3} \leq I_{n-5} \leq \dots \leq I_{4} \leq I_2$ holds if (S1.1), (S1.2), (S1.3), and (S1.4) are fulfilled. By the above argument, the lengths of the intervals are $\ell x^{n-2}, \ell+\ell x, \ell x^2, \ell x^3, \dots, \ell x^{n-2}$, which means that we are now considering $n-1$ intervals. 

	In order to show $I_3 \leq I_5 \leq \dots \leq I_{n-4} \leq I_{n-2} \leq I_1 \leq I_{n-1} \leq I_{n-3} \leq I_{n-5} \leq \dots \leq I_{4} \leq I_2$, we prove four implications:
	\begin{itemize}
		\item (S1.1) implies $I_3 \leq I_5 \leq \dots \leq I_{n-4} \leq I_{n-2}$,
		\item (S1.2) implies $I_{n-2} \leq I_1$,
		\item (S1.3) implies $I_{n-1} \leq I_{n-3} \leq \dots \leq I_{4} \leq I_2$, and
		\item (S1.4) implies $I_4 \leq I_2$.
	\end{itemize}
	The inequality $I_1 \leq I_n$ is correct by the definition of $\tau$. This concludes the proof.
	\begin{itemize}
		\item (S1.1) implies $I_3 \leq I_5 \leq \dots \leq I_{n-4} \leq I_{n-2}$: The argument is the same as in the proof of Lemma~\ref{lem:placevalidngeq6}, where we lower bound $\ell_i$ by $\ell_{n-2} = \ell x^{n-5}$.
		\item (S1.2) implies $I_{n-2} \leq I_1$: The argument is the same as in the proof of Lemma~\ref{lem:placevalidngeq6}, where we substitute $n-1$ by $n-2$.
		\item (S1.3) implies $I_{n-1} \leq I_{n-3} \leq \dots \leq I_{4} \leq I_2$: Similar to the approach for (S1.1), the argument is the same as in the proof of Lemma~\ref{lem:placevalidngeq6}, where we lower bound $\ell_i$ by $\ell_{n-2} = \ell x^{n-5}$.
		\item (S1.4) implies $I_4 \leq I_2$: $I_4 \leq I_2$ is equivalent to $m_2 - \frac{\ell_2}{2} \geq m_4 + \frac{\ell_4}{2}$. We have $\cog{s_1} = \cog{s_3} = - \tau$ and $\cog{s_2} = \cog{s_4} = \tau$. Thus:
			\begin{alignat*}{4}
				&&&m_2 - \frac{\ell_2}{2}\\
				&\;\geq\;&& m_4 + \frac{\ell_4}{2}\\
				\Leftrightarrow &&& \frac{2 \tau \ell_1}{\ell_2} + \tau - \frac{\ell_2}{2}\\
				&\;\geq\; && \frac{2 \tau (\ell_1 + \ell_2 + \ell_3)}{\ell_4} + \tau + \frac{\ell_4}{2}\\
				\Leftrightarrow &&& 2 \tau \left( \frac{\ell_1}{\ell_2} - \frac{\ell_1 + \ell_2 + \ell_3}{\ell_4} \right)\\
				 &\;\geq\; && \frac{1}{2}(\ell_2 + \ell_2)\\
				\Leftrightarrow &&& 2 \tau \left( \frac{\ell x^{n-1}}{\ell + \ell x} - \frac{2 \left( x^{n-1} + 1 + x + x^2 \right)}{\ell x^3} \right)\\
				& \;\geq\; && \frac{1}{2} (x^3 + 1 + x)\\
				\Leftrightarrow &&& 2 \frac{\ell x^{2n-4}(x^2-1)}{4(x^n-1)} \left( \frac{x^{n-1}}{1+x} - \frac{x^{n-1} + 1 + x + x^2}{x^3} \right)
				\\&\; \geq\; && \frac{1}{2} \left( x^3 + 1 + x \right)\\
				\Leftrightarrow &&& \left( x^{2n-2} - x^{2n-4} \right) \left( \frac{x^{n-1}}{1+x} - \frac{x^{n-1} + 1 + x + x^2}{x^3} \right)\\
				 & \;\geq\; && \left( x^3 + 1 + x \right) \left( x^n-1 \right)\\
				\Leftrightarrow &&& \left( x^{2n-2} - x^{2n-4} \right) \left( \frac{x^{n+2} - x^{n-1} - 1 - 2x^{2} - xx^n - 2x - x^3}{x^3 + x^4} \right)\\
				& \;\geq\; && (x^3+1+x)(x^n-1)\\
				\Leftrightarrow &&& x^{2n} (x^{-2} - x^{-4})(x^{n+2} - x^n - x^{n-1} - x^3 - 2x^2 - 2x - 1)\\
				& \;\geq\; &&(x^3 + x + 1)(x^n -1)(x^3+x^4).
			\end{alignat*}
	\end{itemize}
\end{proof}

\begin{lemma}\label{lem:placevalidneq5}
	Let $S = \langle \ell_1,\dots,\ell_n \rangle$ be a heterogeneous system for $n = 5$ and $p = \langle I_1,\dots,I_n \rangle$ the placement that is computed by the algorithm. Then the intervals from $p$ are pairwise disjoint if
	\begin{itemize}
		\item (S1.2): $x^{n+5} + x^{n+2} + x^{n+1} + x^4 + x^2 \geq 2x^{n+4} + x^n + x^5 + x^1$,
		\item (S1.3): $x^{n+5} + x^{n+1} + x^3 + 2x^2 + 1 \geq 2 x^{n+3} + x^{n+2} + x^n + x^5 + x^4$, and
		\item (S1.4): $x^{2n}(x^{-2} - x^{-4})(x^{n+2} - x^n - x^{n-1} - x^3 - 2x^{2} - 2x - 1) \geq (x^3 + x + 1)(x^n -1)(x^3 + x^4)$.
	\end{itemize}
\end{lemma}
\begin{proof}
	Similar to the proof of Lemma~\ref{lem:placevalidngeq7}, we guarantee $I_3 \leq I_5 \leq \dots \leq I_{n-2} \leq I_1 \leq I_{n-1} \leq I_{n-3} \leq \dots \leq I_4 \leq I_2$. As $n =5$, we do not have to deal with $I_3 \leq \dots \leq I_{n-2}$. The proof for (S1.2), (S1.3), and (S1.4) is the same as in the proof of Lemma~\ref{lem:placevalidngeq7}.
\end{proof}

	Now we are ready to give the proof of Lemma~\ref{lem:validplacementheterogenous}:

\begin{proof}[Proof of Lemma~\ref{lem:validplacementheterogenous}]
	
	By combining Lemma~\ref{lem:placevalidngeq6},~\ref{lem:placevalidneq4},~\ref{lem:placevalidngeq7}, and~\ref{lem:placevalidneq5} we obtain that the intervals of $p$ are pairwise disjoint because the Inequations (S1.1), (S1.2), (S1.3), and (S1.4) are fulfilled for $x \geq 2$.
\end{proof}

\section{Conclusion}

We have introduced a new family of problems that seek to balance objects, controlling the variation of their center of gravity during the loading and unloading of the objects. We have provided hardness results and optimal or constant-factor approximation algorithms. 


There are various related challenges. These include sequencing problems 
with multiple loading and unloading stops (which arise in vehicle routing
or tour planning for container ships); variants in which items can be shifted
in a continuous fashion; batch scenarios in which multiple items are loaded
or unloaded at once (making it possible to maintain better balance, but also increasing
the space of possible choices); and
higher-dimensional variants, possibly with inhomogeneous space constraints.
All these are left for future work.

\subsubsection*{Acknowledgements.} We would like to thank anonymous reviewers of the conference abstract for providing helpful comments and suggestions improving the presentation of this paper.

{\small 

\bibliographystyle{abbrv}
\bibliography{refs}

}
\end{document}